\documentclass[11pt]{article}
\pdfoutput=1 
\usepackage[sc]{mathpazo}

\usepackage[margin=1in]{geometry}
\usepackage[english]{babel}
\usepackage[utf8]{inputenc}
\usepackage[compact]{titlesec}

\usepackage{cmap}
\usepackage[T1]{fontenc}
\usepackage{bm}
\usepackage{amsfonts}
\usepackage{amsmath}
\usepackage{amssymb}
\usepackage{amsthm}
\usepackage{float}
\usepackage{graphics}

\usepackage{hyperref}
\usepackage[svgnames]{xcolor}
\hypersetup{colorlinks={true},urlcolor={blue},linkcolor={DarkBlue},citecolor=[named]{DarkGreen},linktoc=all}
\usepackage{natbib}

\usepackage{microtype}
\usepackage[capitalise,nameinlink,noabbrev]{cleveref}

\usepackage{doi}

\usepackage{booktabs} 
\usepackage{amsfonts,amsbsy,amsthm}
\usepackage{thmtools,thm-restate}
\usepackage{enumerate}
\usepackage{bbm}
\usepackage{nicefrac}
\usepackage{wrapfig}
\usepackage[justification = centering]{subcaption}
\usepackage{mathtools}
\usepackage{array,multirow}
\usepackage{fancybox}
\usepackage{multirow}
\usepackage{graphicx}
\usepackage[table,xcdraw,dvipsnames]{xcolor}
\usepackage{comment}
\excludecomment{shortver}
\includecomment{longver}
\usepackage{lscape,booktabs,multirow}
\usepackage{tikz}
\usetikzlibrary{positioning}
\usetikzlibrary{arrows}
\usetikzlibrary{decorations.pathmorphing}
\usetikzlibrary{decorations.markings}
\usetikzlibrary{decorations.pathreplacing}
\usetikzlibrary{shapes.geometric}
\usepackage{pgfplots}
\pgfplotsset{compat=1.17}
\usepackage{thmtools}
\usepackage{subcaption}

\usetikzlibrary{calc,arrows,matrix,decorations.pathreplacing}
\usepackage{paralist}
\usepackage{algorithm}
\usepackage[noend]{algpseudocode}
\usepackage{algorithmicx}

  {\list{}{\leftmargin=#1\rightmargin=#1}\item[]}%
  {\endlist}

\newtheorem{theorem}{Theorem}[section]
\newtheorem{lemma}[theorem]{Lemma}
\newtheorem{definition}[theorem]{Definition}

\newtheorem{observation}[theorem]{Observation}

\newtheorem{remark}[theorem]{Remark}

\newtheorem{example}[theorem]{Example}

\newcommand{\calE}{\mathcal{E}}

\newcommand{\envy}{\ensuremath{\mathsf{envy}}}
\newcommand{\tenvy}{\ensuremath{\mathsf{TotalEnvy}}}
\newcommand{\maxenvy}{\ensuremath{\mathsf{MaxEnvy}}}
\newcommand{\numenvy}{\ensuremath{\mathsf{Envy}}}

\title{Fair Societies: Algorithms for House Allocations}

\author{
  Hadi Hosseini \\
  Pennsylvania State University, USA \\
  \texttt{hadi@psu.edu}
  \and
  Sanjukta Roy \\
  Indian Statistical Institute, Kolkata \\
  \texttt{sanjukta@isical.ac.in}
  \and
  Aditi Sethia \\
  Indian Institute of Science, Bangalore \\
  \texttt{aditisethia@iisc.ac.in}
}

\date{}

\begin{document}

\maketitle 

\begin{abstract}

\emph{House Allocations} concern with matchings involving one-sided preferences, where \emph{houses} serve as a proxy encoding valuable indivisible resources  (e.g. organs, course seats, subsidized public housing units) to be allocated among the agents. Every agent must receive exactly one resource.  
We study algorithmic approaches towards ensuring fairness in such settings. Minimizing the number of envious agents is known to be NP-complete \citep{KMS2021}. We present two tractable approaches to deal with the computational hardness. 
When the agents are presented with an initial allocation of houses, we aim to refine this allocation by reallocating a bounded number of houses to reduce the number of envious agents. We show an efficient algorithm when the agents express preference for a bounded number of houses. Next, we consider single peaked preference domain and present a polynomial time algorithm for finding an allocation that minimize the number of envious agents. We further extend it to satisfy Pareto efficiency. Our former algorithm works for other measures of envy such as total envy, or maximum envy, with
suitable modifications. Finally, we present
an empirical analysis recording the fairness-welfare trade-off
of our algorithms.

\end{abstract}

\section{Introduction}

The problem of house allocation comprises of a set $H$ of $m$ houses to be allocated to a set $N$ of $n$ agents with preferences (cardinal utilities or rankings) such that every agent gets exactly one house. Such one-sided matching problems appear in a wide range of domains, from economic housing markets \citep{SHAPLEY197423} and logistical tasks such as dormitory or course-seat allocations among students \citep{budish2011combinatorial} to critical healthcare domains like allocating donor kidneys among patients \citep{CARAGIANNIS2015}. The rapid integration of automated decision-making processes has made the question of fairness in one-sided matching problems more relevant than ever. 

The gold standard of fairness is \emph{envy-freeness}, where no agent envies the houses allocated to any other agent. In particular, they prefer their allocated house more than any other allocated house. \citet{GSV2019} demonstrated that if such an envy-free allocation exists, it can be computed efficiently. However, such allocations may not even exist, making it imperative to  minimize the unavoidable envy. A few objectives of interest that are studied in the literature (see \citet{Erel2021,KMS2021,Madathil2025}) are as follows: minimizing the sum of envy experienced by all the agents ($\tenvy$), the number of envious agents ($\numenvy$), and the maximum envy experienced by any agent ($\maxenvy$). Interestingly,  the latter two are known to be NP-hard\footnote{If $m\leq n$, the problems are easy (Hosseini et al. \citeyear{Hosseini2024}).} to compute even for binary valuations \citep{KMS2021} and for weak ordinal preferences \citep{Madathil2025}, while the complexity of the former is not known in general.

\subsection{Our Contributions}
 In this work, we focus on minimizing the number of envious agents ($\numenvy$). 
 The computational question of minimizing the number of envious agents is notoriously hard. It is hard to approximate to within a factor of $n^{1-\gamma}$ for any constant $\gamma>0$   where $n$ is the number of agents \citep{KMS2021}, and the exact computation is W[1]-hard with respect to the minimum number of envious agents, even for binary preferences of length three or ordinal preferences \citep{Madathil2025}.
 Towards tractability, we ask: ``given an allocation, is it possible to modify it in order to move closer to a fair allocation minimizing $\numenvy$?" 
\citet{Hosseini2024} proposed to focus on finding fair solutions within the space of all efficient solutions; these are primarily tractable, however, they can be far from optimal fair solutions.
To reach closer to an optimal allocation, we propose an approach that focuses on refining a given allocation via reallocating houses to achieve  fairness, by sequentially expanding the search space. The  benefits are (i) to the best of our knowledge, it is the first tractable approach in the general setting (ii) it allows control over how much reallocation is required or desirable.
 
 


\paragraph{Fixed Parameter Tractability.} Given an allocation, we propose a generic refinement framework via reallocations that enables computation of fairer allocations. 
To this end, we design a \emph{fixed-parameter tractable (FPT)}\footnote{An FPT algorithm with respect to a parameter $\ell$ runs in time $f(\ell) (n+m)^{O(1)}$ for a computable function $f$.} algorithm for the following decision problem: given an allocation $\hat{A}$ and integers $k$ and $q$, can $\hat{A}$ be refined to an allocation $A$ such that $\numenvy$ reduces by $k$ with at most $q$ reallocations (\Cref{thm:fpt_enviousagents}). The parameter is $q+d$, where $d$ is the maximum degree of the associated preference graph. It is relevant to note that our framework is measure-oblivious, and works for any measure of envy -- $\numenvy$, $\tenvy$ or $\maxenvy$ (Theorem \ref{thm:tenvy_maxenvy}).



In a housing market with a large number of houses (i.e., $m$ is large), it is likely that $d$ is small as it is infeasible for an agent to express its preference over all the houses. Agents' approvals of houses are also restricted by locations. A house in upstate wouldn’t be approved by all those who want to live in Manhattan, while a house in Manhattan won’t be liked by someone looking for a peaceful, cheaper, and bigger residence. 
Additionally, certain markets, such as Singapore’s Housing Development Board, also impose quota constraints, wherein each housing project must hold a certain percentage of every major ethnic group \citep{Benabbou18,benabbou2020price}. Consequently, a Chinese applicant may rank a flat in an 87\% Chinese neighborhood lower, recognizing that the ethnic quota renders her ineligible for it. Also, implementing a large number of reallocations (i.e., large value of $q$) can be practically challenging, thereby, justifying the choice of our parameters.

%

\paragraph{Polynomial Time Algorithms.}
Given that minimizing \numenvy{} is hard even for weak ordinal rankings, in search for tractability, we explore single-peaked and single-dipped preferences. These constitute an important preference domain in decision-making problems, which not only model many real-world settings including house allocations and matching markets \citep{BADE201981,beynier2021swap,tamura2023object} but also serve as a tractable realm for many hard problems, albeit not always \citep{10.1145/1562814.1562832}.
In many cases agents' housing preferences are influenced by availability of facilities. For example, elderly individuals may prioritize houses near a hospital, while couples with children might prefer houses closer to a school. The value they assign to a house decreases as its distance from their preferred location increases, producing single peaked preferences. We show that for single-peaked/dipped preferences, an allocation with minimum $\numenvy$ can be found in polynomial time (\Cref{thm:singlepeaks,thm:single-dipped}). Focusing on efficiency, we observe that, although, a Pareto optimal (PO) allocation may not be 
compatible with minimizing \numenvy{}, we can decide the existence in polynomial time (\Cref{thm:singlepeakedPO}). 
We also present structural properties for single-peaked and single-dipped preferences.

\paragraph{Experiments.} Our empirical analysis is done on synthetic data with cardinal preferences. It shows the effects of reallocations on welfare loss and $\numenvy$, and how quickly we can converge (based on the values of $q$) to an optimal allocation given various initial allocations with welfare guarantees such as Nash (geometric mean of agents' utilities) or egalitarian (minimum utility of any agent) welfare. (See Section \ref{sec:cardinal_prelims} for formal definitions.) We average over $100$ instances with $6$ agents and $11$ houses and record (i) the loss in welfare and (ii) the decrease in $\numenvy$ as we increment the number of reallocations  $q$ from $1$ to $n$. We also implement our algorithms for single-peaked/dipped preferences and observe that the welfare loss due to minimizing $\numenvy$ is insignificant when cardinal preferences conform to these structures.


\subsection{Additional Related Work}
House allocations have been studied since early 1970s with various models: existing tenants \citep{SHAPLEY197423}, and new applicants  \citep{HZ1979,ABDULKADIROGLU1999233}. The concept of fairness in house allocations is more recent and was first explored by \citet{BCGHLMW2019,KMS2021,Madathil2025,Hosseini2024}. 
\citet{Erel2021} looked into a relaxed variant where each agent receives at most one house and developed an efficient algorithm to find an envy-free matching of maximum cardinality under binary utilities. \citet{Shende2020StrategyproofAE} examined the interplay between envy-freeness and strategy-proofness. Minimization of various envy measures across all edges in an underlying graph on agents has also been looked at \citep{HPSVV23,10.5555/3635637.3662936,dey2025}. \citet{CHOO2024107103} discussed envy-free house allocations in relation to subsidies.
Adjusting a given allocation for achieving fairness was studied by \citet{He19,friedman2015dynamic,friedman2017controlled} for online fair division settings. 

Single-peaked preferences were first formalized by \citet{2d319a4c-5488-3eef-ab04-29e140a016b3}. A significant literature in social choice has also focused on characterizing single-peaked preferences \citep{e9a63ef1-5488-3dda-8922-497535728793, Elkind2020, PUPPE201855} and they have been studied for domains like voting and electorates \citep{conitzer2007eliciting, 10.1145/1562814.1562832,sprumont1991division} among others. For more details, we refer the reader to the survey of preference restrictions in social choice \citep{elkind2022preference}.

\section{Preliminaries}
\label{sec:prelims}
We denote the set of integers $\{1, 2, \ldots t\}$ by $[t]$. An instance of the house allocation problem with ordinal preferences is given by $\mathcal{I} = (N, H, \succeq)$, where $N$ is a set of $n$ agents, $H$ is a set of $m$ houses, and $\succeq = \{\succeq_{i}|i\in N\}$ is the ranking profile with $\succeq_i$ being the ranking of agent $i$ over (possibly, a subset of ) houses $H$. We assume $m>n$. We use $h\succeq_i h'$ to denote that agent $i$ prefers houses $h$ and $h'$ equally, otherwise we use $h\succ h'$.  

\paragraph{Allocations.}
An allocation $A: N \rightarrow H$ is an injective mapping from the set $N$ of agents to the set $H$ of houses where each agent gets exactly one house. The house allocated to agent $i$ under the allocation $A$ is denoted by $A(i)$. For two allocations $A$ and $A'$, we use $A \Delta  A'$ to denote their symmetric difference.

\paragraph{Fairness.}
    Given an allocation $A$, we say that agent $i$ \emph{envies} agent $j$ if $i$ ranks $j's$ house better than its own house, i.e., $A(j)\succ_i A(i)$. The pairwise envy is defined as $\envy_{i,j} (A):= \mathbb{I}[A(j)\succ_i A(i)]$ where $\mathbb{I}$ denotes the indicator function which is one if the condition is satisfied and zero otherwise. The amount of envy experienced by $i$ is given by $\envy_i(A) : = \sum_{j \neq i} \envy_{i, j}(A)$.
   %
%
An allocation $A$ is \textit{envy-free} (EF) if for every agent $i \in N$ we have 
$\envy_{i} (A) = 0$.
Given an allocation $A$, we denote by $\calE(A)$ the set of all envious agents.
That is, $\calE(A) = \{ i\in N| \envy_{i} (A) > 0\}$. 
We consider the problem of minimizing the number of envious agents ($\numenvy$) in an allocation $A$. Formally, $\numenvy(A) := |\calE(A)|$. The total envy of an allocation $A$ is defined as the sum of the amount of envy experienced by all the agents: $\tenvy(A):= \sum_{i \in \calE(A)} \envy_i(A)$. 
Likewise, the maximum envy of an allocation $A$ is the maximum envy experienced by any agent: $\maxenvy(A):= \max_{i \in \calE(A)} \envy_{i}(A)$. 



\paragraph{Efficiency.} An allocation $A$ is said to be \textit{Pareto dominated} by another allocation $A'$ if at least one agent gets a strictly better house and no agent gets worse-off under $A'$. A Pareto optimal allocation is not Pareto dominated by any allocation. All the missing proofs are deferred to the appendix.

\section{Refining Fairness: A Parameterized Algorithm }
\label{sec:fpt}



\citet{KMS2021} showed that we cannot design a polynomial time algorithm to find a minimum \numenvy{} allocation for binary preferences, unless P = NP. Furthermore, it cannot have an FPT algorithm with respect to the number of envious agents $k$ even for weak ordinal preferences, 
\cite{Madathil2025}. We design a FPT algorithm that, given an initial allocation $\hat{A}$, computes the ``best'' allocation that is at most $q$ reallocations away from $\hat{A}$, parameterized by $q$ and the degree of associated preference graph. 
We first state the main result of this section.

\begin{theorem}
\label{thm:fpt_enviousagents}
    Given an instance $\mathcal{I} = (N, H, \succeq)$ of house allocation, a complete allocation $\hat{A}$, and two positive integers $k$ and $q$, deciding if there is an allocation $A$ such that $\numenvy(A) \leq \numenvy(\hat{A})-k$ and $|A~\Delta~\hat{A}| \leq q$ admits an algorithm that runs in time $O^*(3^{3q(d+1)})$ (randomized) or $O^*(8^{qd\log(3q(d+1)})$ (deterministic), where $d$ is the maximum degree of any vertex in the associated preference graph $G$.
\end{theorem}

\paragraph{Intuition and Challenges.}
Our algorithm starts with an allocation $\hat{A}$. 
To find an allocation with reduced $\numenvy$ from an initial allocation, arguably, the most employed technique is to iteratively eliminate envy cycles. An envy cycle is a directed cycle of agents where each agent envies its out-neighbor. 
To eliminate a cycle, each agent gives her house to the previous agent (who envies her) on the cycle. 
$\numenvy$ can be further reduced by unallocating some  houses and reallocating their occupants to currently unallocated houses. 
These reallocations can be systematically done using $\hat{A}$-alternating paths. An $\hat{A}$-alternating path starts at an unallocated house, continues with a sequence of allocated agent-house pairs, and ends at an allocated house. Thus, changing the allocation along an alternating path unallocates one house and reallocates all the agents on the path. In order to make an envious agent envy-free, we may need to unallocate multiple houses simultaneously (see \Cref{ex:TEnotminE}).
The challenge is to identify the paths and the subset of houses to be allocated (a minimal improvement set as in \Cref{def:improvementset}).
If we need to identify multiple paths that reallocate some $q \in [n]$ agents on the paths,
a brute-force algorithm takes time $m^{O(q)}$ to find the agents and another $O(q^d)$ to identify the houses on the paths, where $d$ is the maximum number of houses in an agent's preference list. We design an algorithm that finds the minimum $\numenvy$ that can be achieved by at most $q$ reallocations from $\hat{A}$, and runs in time $3^{O(qd)}$ which is significantly better when $m$ is large (note $m>n$, more discussion in \Cref{rem:localsearchtime}).
If a few agents require reallocation in $\hat{A}$, then we can converge to an optimal allocation faster (for small $q$). 

\paragraph{Notations.}
The {\em preference graph} $G$ is the bipartite graph between agents and houses, i.e., $G=(N\cup H,E)$ where $E =\{(i,h): \text{agent $i$ ranks house $h$}\}$. 
Let $A$ be an allocation. Then $A$ is a matching in $G$. An \emph{$A$-alternating path} $P$  is a path in $G$ whose edges alternate between allocated and non-allocated edges. We call it an \emph{$A$-alternating cycle} if it starts and ends at the same vertex. 
Then, for an $A$-alternating path/cycle $P$, we define $A\oplus P$ as the allocation obtained from $A$ by removing the edges that appear in both  $A$ and $P$, and adding the edges in $P$ that do not appear in $A$. If $X = \{P_1, \ldots P_s\}$ is a collection of $s$ alternating paths and cycles, then $A \oplus X = A \oplus P_1 \oplus \ldots \oplus P_s$. Suppose that $X$ induces a component $C$ in $G$. Then, we slightly abuse the notation to write $A \oplus C$ instead of $A \oplus X$.
On the other hand, given two allocations $A$ and $\hat{A}$, the symmetric difference, denoted by $A~\Delta~\hat{A}$ is 
a set of $\hat{A}$-alternating path(s) and cycle(s) in $G$ (that are also $A$-alternating path(s) and cycle(s), respectively in $G$). See Example \ref{ex:TEnotminE}. Observe that if $A$ is a complete allocation, then $A \oplus P$ is also a complete allocation where $P$ is a $A$-alternating path/cycle starting at an unallocated house and ending at an allocated house in $A$.



\begin{example}\label{ex:TEnotminE} Consider the following instance with $5$ agents and $8$ houses.
        The allocation $\hat{A} = \{(i_z,h_z)| z \in [5]\}$  is a complete allocation. Note that $\numenvy(\hat{A}) = 5$. There are no envy cycles. Three $\hat{A}$-alternating paths  are $P_1= (h_8,i_1,h_1,i_2,h_2)$, $P_2= (h_7,i_3,h_3,i_4,h_4)$, and $P_3= (h_6,i_5,h_5)$. 
    Furthermore, for each $P_z$, $z \in [3]$, it holds that $\numenvy(\hat{A} \oplus P_z)=\numenvy(\hat{A}) =5$. Let $A = \hat{A} \oplus P_1\oplus P_2 \oplus P_3$. Then, $A\Delta\hat{A}= \{P_1,P_2,P_3\}$, and $A = \{(i_1,h_8),(i_2,h_1),(i_3,h_7), (i_4,h_3), (i_5,h_6)\}$. Surprisingly, $\numenvy(A)= 0$.
    \small{\begin{align*}
        i_1~:~  & h_5 \succeq h_2\succeq h_4 \succ h_8 \succ h_1 \\
        i_2~:~ & h_5 \succeq h_4\succ h_2 \succ h_1 \succ h_8 \\
        i_3~:~ & h_5 \succeq h_2\succ h_4 \succ h_7 \succ h_3 \\
        i_4~:~ & h_5 \succeq h_2\succ h_4 \succ h_3 \succ h_7 \\
        i_5~:~ & h_2 \succeq h_4\succ h_5 \succ h_6 \succ h_1 
    \end{align*}}
    
\end{example}

Therefore, we conclude that envy does not change additively over the alternating paths. That is,
\begin{equation*}
\numenvy(A) < \numenvy(\hat{A}) + \displaystyle\sum_{i=1}^3  \numenvy\text{-drop}(\hat{A},P_i).
\end{equation*}
where $\numenvy\text{-drop}(\hat{A},P_i)=\numenvy(\hat{A}) -\numenvy(\hat{A} \oplus P_i)$ denotes the decrease in number of envious agents in $\hat{A}\oplus P_i$ compared to $\hat{A}$.

\paragraph{Overview of \Cref{alg:minEnvylocal}.} Our algorithm works in two phases. In the first phase, we identify a class $\mathcal{C}$ of subsets of $\hat{A}$-alternating paths and cycles (via randomized coloring or universal set family) such that we can reach the desired allocation using some subsets from $\mathcal{C}$. In the second phase, we begin by deleting the components in $\mathcal{C}$ that are not \emph{
feasible} (Definition \ref{def:feasible_comps}). Then, we identify \emph{minimal improvement sets}, (\Cref{def:improvementset}) of $\mathcal{C}$ such that the number of envious agents decreases by $k$ and no more than $q$ agents' allocations are changed. Towards this, let $C \in \mathcal{C}$ be a feasible component. Then, $C$ contains a set of $\hat{A}$-alternating paths and cycles.
We denote the decrease in the number of envious agents   $\numenvy(\hat{A})-\numenvy(\hat{A}\oplus C)$ as $r_C$. Note that $r_C$ can be negative if the number of envious agents in $\hat{A}\oplus C$ is more than that of $\hat{A}$. 
Additionally, $n_C$ denotes the number of reallocated agents in $\hat{A}\oplus C$. Finally, we solve a knapsack on $\mathcal{C}$ using $r_C$ as profit and $n_C$ as cost to obtain a desired allocation.

\begin{algorithm}[t]
\caption{Reducing $\numenvy$ by with at most $q$ reallocations}
\label{alg:minEnvylocal}
\begin{algorithmic}[1]
\Require $(N, H, \succeq)$, an allocation $\hat{A}$, $k,q \in \mathbb{Z}_{\geq 0}$
\Ensure An allocation $A$ with $\numenvy(A) =\numenvy(\hat{A}) -k$ and $|A~\Delta~\hat{A}|\leq q$.


\State Construct a bipartite graph $G = (N \cup H, E)$ where $(i,h) \in E$ if $i$  ranks $h$ for $i \in N$ and $h \in H$.

\State Let $\chi$ be a coloring of the vertices and edges of $G$ with three colors red, green, and blue uniformly at random.

\State Let $E_B =\{e\in E | \chi (e) = blue\}$, and $\mathcal{C}$ be the connected components of $G -E_B$.

\For{$C \in \mathcal{C}$}
\State if $C$ is \emph{not feasible}, then delete $C$
\EndFor

\For{$C \in \mathcal{C}$}
\State $r_C:= \numenvy(\hat{A}) - \numenvy(\hat{A} \oplus C)$
\State $n_C:=$ number of agents in $C$  
\EndFor

\State Using knapsack algorithm find a subset $X$ of $\mathcal{C}$ such that $\sum_{C \in X} r_C \geq k$ and $\sum_{C \in X} n_C \leq q$ 
\State \Return $A = \hat{A} \oplus X$
\end{algorithmic}
\end{algorithm}

We begin with some  definitions.
\begin{definition}[Dependent Set] Let $\hat{A}$ be a complete allocation. A subset $ T = \{T_1, T_2, \ldots T_t\}$ of $\hat{A}$-alternating paths/cycles  is said to be \textit{dependent} if $ \numenvy(\hat{A} \oplus T_{1} \oplus \ldots \oplus T_{t}) \neq \numenvy(\hat{A}) - \sum_{\ell \in [t]} (\numenvy(\hat{A}) -\numenvy(\hat{A} \oplus T_\ell))$. 
\end{definition}

\begin{definition}
[Minimal Improvement Set]\label{def:improvementset}
Let $T$ be a set of pairwise disjoint  $\hat{A}$-alternating paths/cycles.
A subset $S \subseteq T$ is an \emph{improvement set for $\hat{A}$} if  there exists a positive integer $k$ such that for each subset $S' \subseteq T\setminus S$ and some integer $k'$, it holds that (i) $\numenvy(\hat{A} \oplus S) = \numenvy(\hat{A})-k$, (ii) $\numenvy(\hat{A} \oplus S') = \numenvy(\hat{A}) -k'$, and (iii) $\numenvy(\hat{A} \oplus S \oplus S') = \numenvy(\hat{A}) - (k+k')$.
Further, subset $S \subseteq T$ is a \emph{minimal improvement set} if no subset of $S$ is an improvement set.
\end{definition}


In \Cref{ex:TEnotminE}, set $\{P_1,P_2,P_3\}$ is the only improvement set and thus it is minimal. In Observation \ref{obs:minimalIS}-Lemma \ref{lem:Sisfeasible}, we prove properties of a minimal improvement set that will be helpful in our proof.
Observation \ref{obs:minimalIS}
 follows from the definition.
\begin{observation}
\label{obs:minimalIS}
    Let $S$ be a minimal improvement set for $\hat{A}$. Then,  every path/cycle in $S$ is dependent on at least one other path/cycle in $S$. 
\end{observation}

We show that if $\numenvy(\hat{A})$  is not minimum, then an improvement set exists for the allocation $\hat{A}$.


\begin{restatable}{lemma}{lemimproveset}\label{lem:improveset}
Let $A$ and $\hat{A}$ be two allocations. Then, $\numenvy(\hat{A}) - k =  \numenvy(A)$ for some $k>0$ if and only if $S$ is an improvement set, where $S$ is the set of alternating paths and cycles in $A~\Delta~\hat{A}$.
\end{restatable}

\begin{proof} Suppose $\numenvy(\hat{A}) - k =  \numenvy(A)$ for some $k>0$. 
Then, we have $\hat{A} \oplus T  = A$. 
Thus, we get $\numenvy(\hat{A} \oplus T) = \numenvy(A) = \numenvy(\hat{A}) - k$. Also, note that $T$ is nice by setting $S_1 =T $ and $S_2 = \emptyset$ in the definition.
Therefore, $T$  is an improvement set. 
On the other hand, suppose $T$ is an improvement set. Then, we have $\numenvy(A) < \numenvy(\hat{A})$ as $\hat{A} \oplus T= A $. Therefore, we get  $\numenvy(A) - k = \numenvy(\hat{A})$ for some $k > 0$.
\end{proof}


Next, we prove some properties of an improvement set relating it to the preference graph $G$.
\begin{restatable}{lemma}{sisconnected}
\label{lem:S_connected}
    For any minimal improvement set $S$ for an allocation $\hat{A}$, the graph $G[V(S)]$ is connected.\footnote{$G[V(S)]$ is the preference graph on $V(S)$ (vertices of $S$).}
 \end{restatable}

 \begin{proof} 
 If the set $S$ contains exactly one alternating path / cycle, then we are done. Otherwise, suppose $S = \{T_1, T_2, \ldots, T_s\}$. Since $S$ is minimal, every path/cycle $T_i \in S$ is dependent on at least another path/cycle $T_j \in S$ by \Cref{obs:minimalIS}.

 Now consider two subsets $S_1$ and $S_2$ in $S$ such that $S_1 \cup S_2 = S$. Without loss of generality, we will show that there exists an agent $i \in V(S_1)$ and a house $h \in V(S_2)$ such that $(i, h) \in E(G)$. 
 Suppose not. Then, for each agent $i \in V(S_1)$, we have that $v_i(h) = 0 $ for each house $ h \in V(S_2)$. Then, no agent that appears in $S_1$ values any house that appears in $S_2$. Thus, for each the agents in $S_1$, there is no change in envy due to the allocation or un-allocation of an house in $S_2$. Similarly, we can say that there is no agent in $S_2$ whose envy changes due to allocation or un-allocation of a house in $S_1$. Therefore, let $\numenvy(\hat{A} \oplus S_1) - \numenvy(\hat{A}) = k_1$ and $\numenvy(\hat{A} \oplus S_2) - \numenvy(\hat{A})=k_2$. Then,  we have $\numenvy(\hat{A} \oplus S_1 \oplus S_2) - \numenvy(\hat{A}) = k_1+k_2$. 
 But, by definition, this implies that $S_1$ 
 and $S_2$ are nice subset pairs, contradicting minimality of $S$. 
Hence, we get that $G[V(S)]$ is connected. 
 \end{proof}


\paragraph{Separation of Paths / Cycles.}
A coloring $\chi:V(G) \cup E(G) \rightarrow \{red, green, blue\}$ is a \emph{good} coloring if the following events hold true.

\begin{enumerate}
    \item \label{item:event1} Each agent $i$, house $h$, and the edge $(i,h)$ that appears in a path or cycle in $A~\Delta~\hat{A}$ is colored red. That is, if $(i, h)$ is in $T$, it is colored red. We have, $\chi(i) = \chi(h) = \chi(i,h)=red ~\forall~ i, h, (i, h) \in E(T)$.
    
    \item \label{item:event2} Let $S$ be a minimal improvement set. Then each edge $(i,h)$ in $G$ incident to two vertices of $S$ such that $(i,h) \notin T$ is colored green. That is, $\chi(i, h) = green ~\forall~(i, h) \notin T, i \in S, h \in S$.

    \item \label{item:event3} Let $S$ be a minimal improvement set. Each edge $(i,h)$ in $G$ where either agent $i$ or house $h$ is in $S$ but not both $i$ and $h$ are in $S$, is colored blue. That is, 
    $\chi(i, h) = blue ~\forall~ i \in S, h \notin S ~or ~ i \notin S, h \in S$.
The edges in $G$ that is colored blue is denoted by $E_B$.    
\end{enumerate}

\begin{restatable}{lemma}{probgoodcoloring}
\label{lem:probcoloring}
  Given a coloring $\chi$, the probability that $\chi$ is a good coloring is at least $3^{-3q(d+1)}$.
\end{restatable}

 \begin{proof}
Note that since $|A ~\Delta~ \hat{A}| \leq q$, there are at most $q$
agents in $A$ such that $A(i) \neq \hat{A}(i)$. Then, event (\ref{item:event1}) requires the $q$ agent vertices, $q$ house vertices, and at most $q$ edges between them that appear in $T$ to be colored red, with a probability of $\left(\nicefrac{1}{3}\right)^{3q}$. For the event (\ref{item:event2}), we have at most $2q$ vertices in $S$ and every vertex has degree at most $d$. This gives us at most $qd$ edges which are required to be colored green. This occurs with a
probability of $\left(\nicefrac{1}{3}\right)^{qd}$. Finally, the event (\ref{item:event3})  requires that all the edges (at most $d$) incident to each of the $2q$ vertices in $T$ are colored blue.  This occurs with a probability of $\left(\nicefrac{1}{3}\right)^{2qd}$. Therefore, the probability that all the three events hold is $\left(\nicefrac{1}{3}\right)^{3q + 3qd} = 3^{-3q(1+d)}$.
\end{proof}


Thus, we get a good coloring with high probability if we repeat the above algorithm $3^{3q(d+1)}$ times.

\begin{definition} [Feasible Components] \label{def:feasible_comps} We say that a component $C \in \mathcal{C}$ is \emph{not feasible} if any of the following holds true. (1) There exists a green or blue vertex in $C$. (2) The graph induced by the red vertices and edges in $C$ is not a disjoint union of $\hat{A}$-alternating paths/cycles. (3) There exists an edge in $G$ between two vertices of $C$ that is colored blue. (4) component $C$ is dependent, that is, there exists another component $C' \in \mathcal{C}$ such that  $\numenvy(\hat{A} \oplus C) - \numenvy(\hat{A}) = k_1$, $\numenvy(\hat{A} \oplus C') - \numenvy(\hat{A})=k_2$, and $\numenvy(\hat{A} \oplus C \oplus C') - \numenvy(\hat{A}) < k_1+k_2$. 
A component $C$ is \emph{feasible} otherwise.
\end{definition}

Equipped with the definitions of good coloring and feasible component, we finally prove that each minimal improvement set in $A~\Delta~\hat{A}$ is a feasible connected component in $\mathcal{C}$.

 \begin{restatable}{lemma}{Sisfeasible}     
 \label{lem:Sisfeasible}
     Let $S$ be any minimal improvement set in $A~\Delta~\hat{A}$. Then, in a good coloring, the graph $G[V(S)]$ is a single component in $\mathcal{C}$. Moreover, it is a feasible component. 
 \end{restatable}

  \begin{proof} Suppose $\chi$ is a good coloring. We first show that $G[V(S)]$ is a connected component in $G -E_B$. Then, the edge ($i, h$) between any two vertices in $S$ is colored red if $(i, h) \in E(T)$ and is colored green otherwise. Hence, between any two red-colored paths in $S$, there is a green-colored edge in $E(G) \setminus E_B$. (\Cref{lem:S_connected} ensures that such an edge exists in $G$ and the good coloring $\chi$ ensures that it is colored green.) Therefore, the subset $S$ is a single component in $\mathcal{C}$.
Now we show that the component $C=G[V(S)]$ is feasible by showing it does not satisfy any of the conditions (1) - (4) in the definition of not feasible.
 Note that under any good coloring $\chi$, we have that $C$ does not satisfy (1) as every vertex of $S$ is red. Component $C$ does not satisfy (2) since $S \subseteq T$ and the graph induced by the red vertices and edges in $S$ is a set of $\hat{A}$-alternating path(s)/cycle(s), and its does not satisfy (3) as every edge in $E(G) \cap E(T)$ is colored red and every edge in  $E(G) \setminus E(T)$ is colored green, and no edge in $E(G)$ between two vertices of $C$ is colored blue. Finally, since $C$ is a minimal nice set, it does not satisfy (4). Therefore, by definition, we have that $C$ is a feasible component of $\mathcal{C}$.
 \end{proof}

  We are now ready to argue the correctness of \Cref{thm:fpt_enviousagents}. 

\begin{proof}[Proof Sketch of \Cref{thm:fpt_enviousagents}] Suppose that we have a good coloring $\chi$ for the graph $G$. Then, \Cref{lem:Sisfeasible} ensures that every minimal improvement set is a connected and feasible component in $\mathcal{C}$. We show that once we have the feasible components in $\mathcal{C}$, we need to choose a collection $\mathcal{C}' \subseteq \mathcal{C}$ of components such that $\sum_{C \in \mathcal{C}'} r_C \geq k$ and $\sum_{C \in \mathcal{C}'} n_C \leq q$. 
Note that this is exactly the classical knapsack problem with $q$ as the maximum admissible weight and $k$ as the minimum required profit. In particular, the input to the knapsack problem is $\langle \mathcal{C}, r_1, r_2,\ldots, r_{|\mathcal{C}|}, n_1, n_2, \ldots n_{|\mathcal{C}|}, q, k  \rangle$ and goal is to decide if there is a subset $X \subseteq \mathcal{C}$ such that $\sum_{C \in X} n_C \leq q$ and $\sum_{C \in X}r_C \geq k$. We defer the proof of equivalence to the appendix.
 The knapsack problem can be solved in time $O(|\mathcal{C}|q) = O(nq)$~\citep{DBLP:books/daglib/0010031}. 
For the randomized algorithm, a good coloring is obtained with high probability by repeating the step $3^{3q(d+1)})$ times.
Thus, total time taken is $O((n+m)^2 \cdot 3^{3q(d+1)} )$. 
A derandomization of \Cref{alg:minEnvylocal} and its run time analysis is presented in \Cref{sec:derandomization}.
\end{proof}

We now show that Algorithm \ref{alg:minEnvylocal} is oblivious to the measure of envy under consideration. It can be suitably adapted for total envy or maximum envy of an allocation. In fact, it works for ``\textit{cardinal preferences}'' as well by only modifying Line 1 of \Cref{alg:minEnvylocal} to define the edge set of preference graph $G$ as follows: $(i,h) \in E$ if agent $i$ has a non-zero, positive value for house $h$ for $i \in N$ and $h \in H$. 

\begin{theorem}\label{thm:tenvy_maxenvy}
Given an instance $\mathcal{I} = (N, H, V)$ of house allocation, a complete allocation $\hat{A}$, and two positive integers $k$ and $q$, deciding if there is an allocation $A$ such that $\tenvy(A) \leq \tenvy(\hat{A})-k$ (or $\maxenvy(A) \leq \maxenvy(\hat{A})-k$) and $|A~\Delta~\hat{A}| \leq q$ admits a fixed-parameter tractable algorithm parameterized by $q$ and $d$, where $d$ is the maximum degree of ant vertex in the associated preference graph $G$.   
\end{theorem}

\section{Restricted Domain: Efficient Algorithms}
\label{sec:restricted}

In this section, we present efficient algorithms for minimizing $\numenvy$ when the rankings over $H$ are complete, strict, and single-peaked/dipped. 

\subsection{Single-Peaked Preferences}

We first present the definition. We say $h$ is a peak house for agent $i$, denoted as $peak(i)$ if $h$ is the first-ranked house for $i$. That is, $h \succ_i h' $ for each $ h' \neq h$.



\begin{definition}[Single-Peaked Preferences]
A preference ranking $\succ_i$ is \emph{single-peaked} with
respect to an ordering~$\rhd$ of houses $H$ if for every pair of houses $h,h' \in H$, we have that if $h \rhd h' \rhd~peak(i)$ or $peak(i) \rhd h' \rhd h$, then $h' \succ_i h$. A preference profile $\succ$ is single-peaked if there exists an ordering $\rhd$ over $H$ such that $\succ_i$ is single-peaked with respect to $\rhd$ for every agent $i \in N$ (see \Cref{fig:spdefinitions}). 

\end{definition}

Intuitively, in single-peaked preferences, as an agent moves away from his favorite house $peak(i)$ in the ordering $\rhd$ in any direction, left or right, the houses become less preferable for her. 

We now state the main result of this section, define the notations, and structural properties that will be helpful in the proof.

\begin{restatable}{theorem}{singlepeaks}
\label{thm:singlepeaks}
    Given a single-peaked instance $\mathcal{I} = (N, H, \succ, \rhd)$, minimizing the number of envious agents admits a polynomial time algorithm with run time $O(|H|^2)$.
\end{restatable}



\paragraph{Notations.}
 For a house $h$, we denote the set of agents who prefer it to all other houses as $base(h)$. That is, $base(h) = \{ i \in N~|~ h \succ_i h' \text{ for all } h' \neq h\} = \{i \in N~|~ h = peak(i)\}$. Suppose that the set of rankings $\succ$ are single-peaked with respect to the ordering $\vartriangleright$ over the houses $(h_1 \vartriangleright h_2 \vartriangleright \ldots \vartriangleright h_m)$. The interval $[h_i, h_j)$ denotes the set of houses $h_i \vartriangleright \ldots \vartriangleright h_{j-1}$. We say that a house $h$ is a \emph{shared peak} if it is the most preferred house of more than one agent, that is, $|base(h)|>1$. Otherwise, if $base(h) = 1$, we say it is an \emph{individual peak}. We say that a house $h$ is \emph{non-wastefully} allocated if it is allocated to an agent in $base(h)$, otherwise it is allocated wastefully. We define the \emph{span} of a peak house $h$, denoted by $span(h)$, as the set of houses that are identically ranked by {\em all} the agents in $base(h)$, starting from their first ranked house. If a house $h$ is an individual peak, then we say $span(h) = \emptyset$.

\begin{example}
\label{ex:spdefinitions}
Consider the instance in \Cref{fig:spdefinitions} with four agents and the following rankings. 

\begin{table}[!h]
    \footnotesize
    \centering
    \begin{tabular}{ll}
    $i_1:$ & $\mathbf{h_2} \succ h_1 \succ h_3 \succ h_4 \succ h_5 \succ h_6 \succ h_7 $  \\
    $i_2:$ & $h_4 \succ h_5 \succ h_6 \succ \mathbf{h_3} \succ h_2 \succ h_1 \succ h_7$ \\
    $i_3:$ & $h_4 \succ h_5 \succ h_6 \succ h_3 \succ h_7 \succ h_2 \succ \mathbf{h_1} $ \\
    $i_4:$ & $h_4 \succ h_5 \succ h_6 \succ \mathbf{h_7} \succ h_3 \succ h_2 \succ h_1$ 
    \end{tabular}
    \label{tab:example1}
\end{table}

It is a single-peaked instance with respect to the ordering $\rhd:= h_1 \rhd h_2 \rhd h_3 \rhd h_4 \rhd h_5 \rhd h_6 \rhd h_7$. The house $h_4$ is a shared peak, and $h_2$ is an individual peak. Notice that peak($i_1$) = $h_2$ and peak($i_2$) = peak($i_3$) = peak($i_4$) = $h_4$. Consequently, base($h_4$) = $\{i_2, i_3, i_4\}$. Note that $|span(h_4)| = |\{h_4,h_5,h_6\}|= 3$, as these are the top houses identically ranked by all the agents in $base(h_4)$. Also, $|span(h_1)| = 0$. An important and helpful observation is the following. If at least two agents from $\{i_2, i_3, i_4\}$ were to be envy-free in any allocation, then not only the peak house $h_4$ would have to remain unallocated, but all the houses in $span(h_4)$ must also remain unallocated under any complete allocation. The allocation of $h_2$, $h_3$, $h_1$, and $h_7$ to the four agents, respectively, is the one minimizing $\numenvy$ with exactly one envious agent, namely, $i_3$.


\begin{figure}[t]
    \centering
    \resizebox{0.6\linewidth}{!}{%
        \begin{tikzpicture}
        \begin{axis}[
            width=10cm,
            height=7cm,
            axis x line=bottom,
            axis y line=none,
            xtick={1,2,3,4,5,6,7},
            xticklabels={$h_1$, $h_2$, $h_3$, $h_4$, $h_5$, $h_6$, $h_7$},
            ytick=\empty,
            tick style={line width=0.5pt},
            legend pos=north east,
            legend cell align={left},
            legend style={draw=none}, 
            grid=both,                
            grid style={gray!30}
        ]

        \addplot[RoyalBlue, thick] coordinates {(1,5) (2,6) (3,5) (4,4) (5,3) (6,2) (7,1)};
        \addlegendentry{$i_1$}

        \addplot[orange, thick] coordinates {(1,1) (2,2) (3,3) (4,6) (5,5) (6,4) (7,0)};
        \addlegendentry{$i_2$}

        \addplot[ForestGreen, thick, dash pattern=on 8pt off 4pt] coordinates {(1,2) (2,3) (3,4) (4,6) (5,5) (6,4) (7,2.4)};
        \addlegendentry{$i_3$}

        \addplot[BrickRed, thick, dash pattern=on 10pt off 3pt on 1pt off 3pt] coordinates {(1,0) (2,1) (3,2) (4,6) (5,5) (6,4) (7,3)};
        \addlegendentry{$i_4$}

        \end{axis}
        \end{tikzpicture}
    }
    \caption{Single-Peaked preferences}
    \label{fig:spdefinitions}
\end{figure}
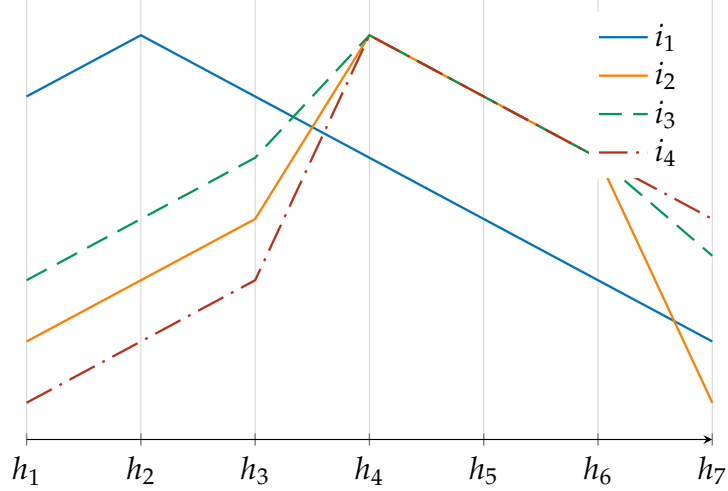

\end{example}

We now present a series of structural results. If a shared peak $h$ is assigned to one agent, it leads to envy among the other agents, with at least $|base(h)| - 1$ envious agents. On the other hand, we show by the following claim that even if $h$ is not assigned, at least $|base(h)|-2$ agents are bound to be envious under any allocation.

\begin{lemma}
\label{claim:atmosttwo} Let $h$ be a shared peak. Then, at most $2$ agents from the set $base(h)$ can be envy-free under any allocation. 
\end{lemma}

\begin{proof} Consider an allocation $A$. If the house $h$ is allocated wastefully, then all the $base(h)$ agents are envious, no matter which house they receive in $
A$. If $h$ is allocated non-wastefully to an agent, say $i$, then $i$ is always envy-free in any completion of this allocation, as she receives her first-ranked house. But, all other $|base(h)|-1$ agents are envious of  $i$. If house $h$ is not allocated in $A$, then we prove at most two agents in $base(h)$ are envy-free. Let a single-peaked axis for the preferences be denoted by $\vartriangleright= h_1\vartriangleright\dots\vartriangleright h_m$. Consider three agents $i_1, i_2$, and $i_3$ from $base(h)$. Then, at least $2$ of these agents are allocated to houses from the interval either $[h_1, h)$ or $(h, h_m]$ on the axis $\vartriangleright$. WLOG, we assume that $A(i_1) = h_j$ and $A(i_2) = h_l$ such that $\{h_j, h_l\} \in [h_1, h)$ and $l<j$. Then, since $h$ is a peak for both $i_1$ and $i_2$, by the structure of the rankings, it holds that both agents $i_1$ and $i_2$ have the (partial) ranking $h_1 \prec h_l \prec h_i \prec h$. Then, agent $i_2$ envies $i_1$. Thus, at most one envy-free agent can be assigned to each interval $[h_1, h)$ and $(h, h_m]$. Therefore, at most two agents can be envy-free from the set $base(h)$. Furthermore, the  houses allocated to the two agents lies on either side of $h$ in a single peaked axis.
\end{proof}

We now proceed to show another interesting structural claim that helps us to allocate the individual peaks non-wastefully. 

\begin{restatable}{lemma}{individualpeaks}
\label{lemma:individualpeaks}There exists an allocation with the minimum number of envious agents where all individual peaks are allocated, and they are allocated non-wastefully.
\end{restatable}

\begin{proof}  Let house $h$ be an individual peak. Suppose that $h$ remains unallocated under an allocation $A^*$. Consider the following cases:

\begin{enumerate}
\item $h \in \text{span}(h_i)$ for some shared peak $h_i$. Then $h$ remains unallocated and at most two agents from the set of agents peak$(h_i)$, say $i_1$ and $i_2$, are envy-free under $A^*$ (this forces span($h_i$) and hence, $h$ to remain unallocated under $A^*$ ). Now consider the reallocation where $a$ gets $h$ and $i_1$ gets $h_i$ (where $i_1 \in base(h_i)$). Then both $i$ and $i_1$ are envy-free under this reallocation. The number of envy-free agents under this reallocation remains the same as in $A^*$. Indeed, at most two out of $\{i, i_1, i_2\}$ can be made envy-free under any allocation (by \Cref{claim:atmosttwo}). This settles our claim in this case.

\item $h \notin \text{span}(h_i)$ for any other peak $h_i$. First, suppose that $a$ is an envy-free agent under $A^*$. This implies that every house that $i$ ranks better than $A^*(i)$ (including $h$) remains unallocated under $A^*$. We claim that on the reallocation of $h$ to $i$, no new envious agent is created. Suppose not. Say $i'$ is an agent who was previously envy-free but becomes envious on the allocation of $h$. Since $i'$ was envy-free previously, we can say that $A^*(i') >_{i'} A^*(i)$. If $A^*(i')$ lies to the left of $A^*(i)$, then by structure of the valuations, $A^*(i') >_{i'} h$ and therefore, $i'$ can't be envious of the allocation of $h$. On the other hand, if $A^*(i')$ lies to the right of $A^*(i)$, then it must be that $A^*(i')$ also lies to the right of $h$. Else, $a$ will prefer $A^*(i')$ more than $A^*(i)$, leading her to be envious. Therefore, it must be the case that $A^*(i')$ is to the right of $h$
 and hence by structure of the valuations, $A^*(i) <_{i'} h <_{i'} A^*(i')$. Therefore, $i'$ can't be envious of the allocation of $h$.

Now suppose that $i$ was envious under $A^*$.
Reallocating $h$ to $i$ makes her envy-free. We will first argue that there are at most two agents, say $i_1$ and $i_2$ who become newly envious of the allocation of $h$. Suppose the peaks of $i_1$ and $i_2$ lie to the left and right side of $h$ respectively. Since both $i_1$ and $i_2$ have become newly envious, so $A^*(i_1)$ and $A^*(i_2)$ can not be their respective peak houses.  Since $i_1$ and $i_2$ are envy-free under $A^*$, therefore, $peak(i_1) >_{i_1} h >_{i_1} A^*(i_1) >_{i_1} A^*(i)$ and $peak(i_2) >_{i_2} h >_{i_2} A^*(i_2) >_{i_2} A^*(i)$. And all houses between  $A^*(i_1)$ and $peak(i_1)$ are unallocated. Similarly, all houses between $A^*(i_2)$ and $peak(i_2)$ are unallocated. Since $h$ is a common house that lies both  between $A^*(i_1)$ and peak($i_1$) and between  $A^*(i_2)$ and peak($i_2$), we have that all the houses between $A^*(i_1)$ and $A^*(i_2)$ are unallocated (see \Cref{fig:lemma2}). Now, consider any other agent $i_3$ who is envy-free under $A^*$. We claim that she can't be envious of the allocation of $h$. First suppose that $A^*(i_3)$ lies to the left of $A^*(i_1)$. Then it must be the case that $peak(i_3)$ also lies to the left of $A^*(i_1)$ because $i_3$ is envy-free under $A^*$. This means that if $i_3$ is envy-free of the allocation of $A^*(i_1)$, she remains envy-free of the allocation of $h$ as well (by the structure of the valuations). Second, if $A^*(i_3)$ lies to the right of $A^*(i_2)$. Then it must be the case that $peak(i_3)$ also lies to the right of $A^*(i_2)$ which means that if $i_3$ is envy-free of the allocation of $i_2$, she remains envy-free of the allocation of $h$ as well (again by the structure of the valuations).

So now we have that $i_1$ and $i_2$ are the only two agents that can potentially become envious of the allocation of $h$. Note that $peak(i_1)$ and/or $peak(i_2)$ can not be a resolved shared peak under $A^*$. Otherwise, $h$ must be in one of the spans, which contradicts the assumption of this case. Therefore, in the optimal allocation $A^*$, at most one agent from $base(peak(i_1))$ and at most one agent from $base(peak(i_2))$ is envy-free.

\begin{figure}[htbp]
    \centering
    \begin{tikzpicture}[scale=1]

        \draw[->] (-0.6,0) -- (7.2,0);

        \foreach \x in {-0.6, 0, 1, 2, 3.25, 4.5, 6.2, 7}
            \draw[dashed,gray!30,thin] (\x,0) -- (\x,3.9);

        \draw[BrickRed, thick] 
          (-0.5,0.1) -- (0,0.8) -- (2,3.3) -- (3.25,3.9) -- (4.5,3.3) -- (6.5,0.5);

        \draw[ForestGreen, thick, dashed]
            (-0.5,1.2) -- (1,2.5) -- (2,3.3) -- (3.25,3) -- (6.7, 1.0);

        \draw[cyan, dotted, ultra thick]
            (-0.5,0.2) -- (1,1.2) -- (3.25,3) -- (4.5,3.3) -- (6.7,1.5);

        \foreach \x/\y in {0/0.8, 1/2.5, 2/3.3, 4.5/3.3, 6.2/1.9}
            \filldraw[black] (\x,\y) circle (1.3pt);

        \node[below, yshift=-2pt] at (-0.6,0) { };
        \node[below, yshift=-2pt] at (0,0)     {$A^{\star}(i)$};
        \node[below, yshift=-2pt] at (1,0)     {$A^{\star}(i_1)$};
        \node[below, yshift=-2pt] at (2.3,0)     {$\mathrm{peak}(i_1)$};
        \node[below, yshift=-2pt] at (3.25,0)  {$h$};
        \node[below, yshift=-2pt] at (4.5,0)   {$\mathrm{peak}(i_2)$};
        \node[below, yshift=-2pt] at (6.2,0)   {$A^{\star}(i_2)$};
        \node[below, yshift=-2pt] at (7,0)   { };

        \node[BrickRed, above, yshift=4pt] at (3.25,3.9) {$A(i)$};
        \node[BrickRed, above left, xshift=-14pt, yshift=4pt] at (2.5,3.3) {$A(i_1)$};
        \node[BrickRed, above right, xshift=10pt, yshift=4pt] at (4.2,3.3) {$A(i_2)$};

        \node[black, left, xshift=2pt] at (0.8,2.6) {$A^{\star}(i_1)$};
        \node[black, right, xshift=2pt] at (6.2,2.0) {$A^{\star}(i_2)$};
        \node[black, left,  xshift=-4pt] at (0,0.8) {$A^{\star}(i)$};

        \begin{scope}[shift={(-.7,4.3)}]
            \draw[BrickRed, thick] (0.2,-0.2) -- (0.8,-0.2);
            \node[right] at (0.9,-0.2) {$i$};
            \draw[ForestGreen, thick, dashed] (0.2,-0.6) -- (0.8,-0.6);
            \node[right] at (0.9,-0.6) {$i_1$};
            \draw[cyan, dotted, ultra thick] (0.2,-1.0) -- (0.8,-1.0);
            \node[right] at (0.9,-1.0) {$i_2$};
        \end{scope}

    \end{tikzpicture}
    \caption{A schematic of Case 2 in the proof of \Cref{lemma:individualpeaks}.}
    \label{fig:lemma2}
\end{figure}
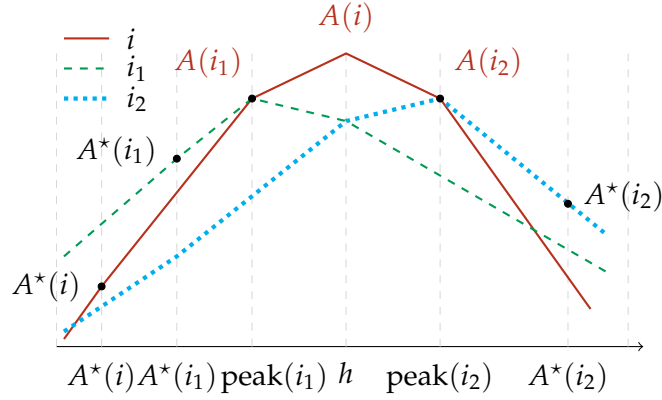

Based on this, we can now propose the following reallocation: $h$ to $i$, $peak(i_1)$ to $i_1$ and $peak(i_2)$ to $i_2$. It is easy to see that $\{i, i_1, i_2\}$ become newly envy-free as each of them now gets her peak house. Also, no other agent becomes newly envious of this reallocation. Indeed, if there is such an agent, say $i'$, then $A^*(i')$ must lie to either left of $A^*(i_1)$ or to the right of $A^*(i_2)$. Since they are envy-free under $A^*$, their respective peaks must also be on the same side as their allocated houses. This means that if they are envy-free of the allocation of $A^*(i_1)$ and $A^*(i_2)$ respectively, they do not become envious of the allocation of $peak(i_1), ~peak(i_2)$ and $h$, again by the structure of the valuations. Therefore, $\{i, i_1, i_2\}$ are newly envy-free in the reallocation, without creating any other envious agents. But under $A^*$, only $\{i_1, i_2\}$ were envy-free. This implies that $A^*$ was not optimal to begin with. This settles our claim.
\end{enumerate}

Now suppose that the individual peak, say $h$ is allocated wastefully to some agent $i^\prime$ under $A^*$. Let $i$ be the unique agent in the set $base(h)$. Then, $i$ is definitely envious. If $i^\prime$ is also an envious agent, then we can
re-allocate $h$ to $i$, which reduces the number of envious agents and contradicts the optimality of $A^*$. Therefore, $i^\prime$ must be an envy-free agent. Then, all the houses that $i^\prime$ values more than $h$ must have remained unallocated. In particular, the peak house of $i^\prime$, say $h^\prime$, ($h^\prime \neq h$) must have been unallocated. If $h^\prime$ was a shared peak, then it must have been a resolved shared peak (since $h^\prime$ remains unallocated) and hence at most two agents would be envy-free from $base(h^\prime)$. Notice that $i^\prime$ is one of them and, say $i^{\prime\prime}$ is the other. Then, reallocating $h'$ to $i'$ and $h$ to $i$ converts $i$ to an envy-free agent and makes $i^{\prime\prime}$ envious, and does not generate any new envy. Otherwise, if $h^\prime$ was an individual peak, then again the reallocation of $h$ to $i$ and $h^\prime$ to $i^\prime$ gives us our desired allocation. This settles our claim.
\end{proof}

Let the number of individual and shared peaks be $p_I$ and $p_S$ respectively. Then any allocation can have at least $p_I + p_S$ envy-free agents, just by allocating the peaks non-wastefully and completing the allocation in an arbitrary manner. Moreover, by \Cref{claim:atmosttwo}, no allocation can have more than $p_I + 2 \cdot p_S$ envy-free agents. This establishes the following result.

 \begin{algorithm}[t]
\caption{Minimize $\numenvy$ for Single-Peaked Preferences}
\label{alg:numberenvy}
\begin{algorithmic}[1]
\Require $(N, H, \succ)$ and a single-peak axis $\vartriangleright$
\Ensure Allocation $A$ that minimizes $\numenvy$

\noindent
{\it\textbf{Base cases}:}
\For{$h \in p_I$}
\do $A(base(h)) = h$ 

\Comment{Allocate $p_I$ non-wastefully.} 
\EndFor
\For{$h_j, h_l ~\text{s.t.}~ h_j \in span(h_l) ~\text{and}~ h_l \in span(h_j)$} 
 \do $A(i) = h_j$,  $A(i^\prime) = h_l$ s.t. $i \in base(h_j)$, $i^\prime\in base(h_l)$ 
  \EndFor

\For{$h_j, h_l ~\text{s.t.}~ h_j \in span(h_l) ~\text{but}~ h_l \notin span(h_j)$} 
\do $A(i) = h_l ~\text{for some}~ i \in base(h_l)$ 
 \EndFor

\noindent
{\it\textbf{Greedy resolve}:}
\State $S\leftarrow$ Set of remaining unallocated shared peaks.

\State Order the houses in $S$ as $h \preceq h'$ if $|span(h)| \leq |span(h')|$
Say, $\{h_{z_1}, h_{z_2}, \ldots h_{z_{|S|}}\}$ is the ordering.

\For{$j \in [S]$}
$m^\prime, n^\prime$ = number of unallocated houses \& agents under $A$

\If{$m^\prime - |span(h_{z_j})| \geq  n^\prime$}
Resolve $h_{z_j}$ 
\& $U \leftarrow span(h_{z_j})$
\Else   
~Allocate $\{h_{z_j}, h_{{z_j}+1}, \ldots h_{z_S} \}$ non-wastefully.
\EndIf
\EndFor

\State Order the remaining agents and let each agent choose its highest ranked house among the remaining houses, except from $U$.
\State Output $A$.
\end{algorithmic}
\end{algorithm}

\begin{lemma}
\label{lem:numEF}
Let $|EF(A^\star)|$ be the number of envy-free agents under any optimal allocation $A^\star$. Then,
$p_I + p_S \leq |EF(A^\star)| \leq  p_I +2 \cdot p_S$.
\end{lemma}

The following is a generalization of \Cref{claim:atmosttwo}.

\begin{restatable}{lemma}{overlappingspans}
\label{lem:atmostk+1}
    Let $\{h_1, h_2, \ldots h_k\}$ be the set of $k$ shared peaks such that $span(h_j) ~\cap~ span(h_l) \neq \emptyset$ for any $j, l \in [k]$. Then at least $k$ and at most $k+1$ agents from the set $\bigcup_{j \in [k]} base(h_j)$ are envy-free under any optimal allocation.
\end{restatable}

\begin{proof}
Consider a non-wasteful allocation of the $k$ shared peaks among $k$ agents in the set $\bigcup_{i \in t} base(h_i)$. Clearly, these $k$ agents are envy-free in any complete allocation, as each of them receives their favorite house. This allocation makes at least $k$ agents envy-free. 
Now consider an allocation where at least one span is resolved, say span($h_1$). This means that $h_1$ and span($h_1$) remain unallocated and consequently, $2$ agents from $base(h_1)$ are made envy-free by the allocation of the houses, say $h_1^1$ and $h_1^2$. Now consider any other overlapping span, say span($h_i$). Then $span(h_i)  \cap span(h_1) \neq \emptyset$. Then, by the structure of the rankings, either $h_1^1$ or $h_1^2$ must belong to the span($h_i$). This implies that once span($h_1$) is resolved (that is, $h_1^1$ and $h_1^2$ are allocated), then span($h_i$) can not be resolved (at most one agent from peak$(h_i)$ can be made envy-free). Since the choice of $h_i$ was arbitrary, this holds for every other overlapping span with span($h_1$). Therefore, at most $k+1$ agents can be envy-free from the set $\bigcup_{i \in t} base(h_i)$. 
\end{proof}

A shared peak $h$ is said to be \emph{resolved} under an allocation $A$ if $span(h)$ remains unallocated and as a result exactly two agents from the set $base(h)$ become envy-free under $A$, specifically, the agents from $base(h)$ who get their $span(h)+1$ ranked house (\Cref{claim:atmosttwo}).


\paragraph{Overview of \Cref{alg:numberenvy}.} In the light of \Cref{lemma:individualpeaks}, we first allocate all the individual peaks non-wastefully. For a pair of shared peaks $h_j$ and $h_l$ such that $h_j \in span(h_l)$ and $h_l \in span(h_j)$, we allocate $h_j$ and $h_l$ non-wastefully. Otherwise if $h_j \in span(h_l)$ but $h_l \notin span(h_j)$, then
in the light of \Cref{lem:atmostk+1}, at most $3$ agents from $base(h_j) \cup base(h_l)$ can be envy-free, and to that end, we allocate $h_l$ non-wastefully. Now what remains are the shared peaks, possibly with overlapping spans. We resolve these remaining peaks in a greedy manner, by choosing a peak with the minimum span size in each step. In each step, we resolve a peak and set some unassigned houses as unavailable. We resolve the peaks as long as the number of unallocated agents is less than the number of available unallocated houses.
Finally, we complete the allocation by letting the remaining unallocated agents choose an available unallocated house, one by one. We are now ready to present the proof idea of the main result.

\begin{proof}[Proof Idea of \Cref{thm:singlepeaks}] 
Let $EF(A)$ denote the set of envy-free agents in an allocation $A$.
Let $A$ be the output of \Cref{alg:numberenvy} and $A^*$ be an optimal allocation.  Clearly, $|EF(A^*)| \geq |EF(A)|$. To prove the correctness of \Cref{alg:numberenvy}, we show that $|EF(A^*)| = |EF(A)|$. To this end, we show that $|EF(A^*) \setminus EF(A)| = |EF(A) \setminus EF(A^*)|$ by analyzing cases based on if a peak is in the span of another.
We defer that to the appendix.
\end{proof}

It is relevant to note that, unlike Algorithm \ref{alg:minEnvylocal}, Algorithm \ref{alg:numberenvy} does not extend to other envy measures like $\tenvy$. This is because the structural claim of Lemma \ref{claim:atmosttwo} is specific to $\numenvy$. For examples, consider Example \ref{ex:spdefinitions}. The allocation $A$ such that $A(i_1) = h_2, A(i_2)=h_4, A(i_3)=h_5, A(i_4)=h_7$ minimizes $\tenvy$ which entails allocating houses from $span(h_4)$. The highlighted allocation where $span(h_4)$ remains unallocated is suboptimal for minimizing $\tenvy$.

\subsection{Single-Dipped Preferences}
We now present out results for single-dipped preferences.
Analogous to single-peaked case, under single-dipped preferences, there is an ordering $\rhd$ on the houses such that for any agent $i$, there is a least preferred house $h$, called a \emph{$dip(i)$}, and she prefers the houses better as she moves away from $h$ in either direction with respect to the ordering $\rhd$. Formally, 

 \begin{definition}[Single-Dipped Preferences]
A preference ranking $\succ_i$ is \emph{single-dipped} with
respect to an ordering~$\rhd$ of houses $H$ if for every pair of houses $h,h' \in H$, $h \rhd h' \rhd dip(i)$ or $h \rhd h' \rhd dip(i)$, implies $h \prec_i h'$. A preference profile $\succ$ is single-dipped if there exists an ordering $\rhd$ over $H$ such that $\succ_i$ is single-dipped with respect to $\rhd$ for every $i \in N$.
\end{definition}

We now present the main result for such preferences.

\begin{restatable}{theorem}{singledip}
\label{thm:single-dipped} Given a single-dipped instance $\mathcal{I} = (N, H, \succ, \rhd)$ of house allocation, minimizing the number of envious agents admits a linear time algorithm with runtime $O(m)$.
\end{restatable}


\begin{algorithm}[t]
\caption{Minimize $\numenvy$ for Single-Dipped Preferences}
\label{alg:numberenvySD}
\begin{algorithmic}[1]
\Require $\{N, H, \succ, \vartriangleright \}$
\Ensure An allocation $A$ that minimizes the number of envious agents
\State $S_1 := \{h \in [m]: rank_i(h) = 1~\text{for}~i \in [n]\}$
\If{$|S_1| > 1$}
\State $A(i_1) = h_1$ s.t. $h_1 \in S_1 ~\&~i_1 \in base(h_1)$
\State $A(i_2) = h_2$ s.t. $h_2 \in S_1 ~\&~i_2 \in base(h_2), h_2 \neq h_1$
\State Order the remaining agents and let each agent choose its highest ranked house among the remaining houses.
\Else~$|S_1|=1$, say $h \in S_1$
\If{$m - span(h) \geq n$}
\State $S_{span(h)+1} := \{h': rank_i(h') = |span(h)|+1~\text{f.s.}~i \in [n]\}$
\State Let $h_1, h_2 \in S_{span(h)+1}, h_1 \neq h_2$
\State $A(i_1) = h_1$ such that $rank_{i_1}(h_1) = |span(h)|+1$
\State $A(i_2) = h_2$ such that  $rank_{i_2}(h_2) = |span(h)|+1$
\State $U \leftarrow span(h)$ 
\State Repeat Step 5 on the houses in $M \setminus U$
\Else~$m - |span(h)| < n$
\State $A(i) = h$ for some $i \in base(h)$
\State Repeat Step $5$
\EndIf
\EndIf
\State Output $A$
\end{algorithmic}
\end{algorithm}

We first present the following interesting structural claim.

\begin{restatable}{lemma}{atmosttwodipped}
\label{claim:atmost2}
When the preferences are single-dipped, at most two agents can be envy-free under any complete allocation.
\end{restatable} 

\begin{proof}
Suppose $\{h_1, h_2, \ldots h_m\}$ is the ordering of the houses with respect to which the preferences are single-dipped. Notice that for every agent, either $h_1$ or $h_m$ is the first-ranked house. If both these houses are allocated under an allocation $A$, say to agents $i_1$ and $i_2$, such that $rank_{i_1}(h_1)=1$ and
 $rank_{i_2}(h_m)=1$, then it is easy to see that both $i_1$ and $i_2$ are envy-free. Notice that these are the only envy-free agents since any other agent $i$ is already envious of the allocation of either $h_1$ or $h_m$. This settles our claim in this case.
 Now suppose $h_j$ is the first house in the ordering that is allocated (to $i_1$) and $h_l$ is the last one, allocated to $i_2$, such that both $i_1$ and $i_2$ are envy-free. Consider any other agent $i$. If the dip of $i$ lies to the left of $h_j$, then since $A(i) \in (h_j, h_l)$, $i$ is envious of the allocation of $h_l$. If the dip of $i$ lies to between $h_j$ and $h_l$, then again $A(i) \in (h_j, h_l)$, and $i$ would be envious of both $h_j$ and $h_l$. Lastly, if the dip of $i$ lies to the right of $h_l$, then $i$ is envious of the allocation of $h_j$. This settles the claim.
\end{proof}

\paragraph{Overview of \Cref{alg:numberenvySD}.} Based on \Cref{claim:atmost2}, \Cref{alg:numberenvySD} (in the appendix) works as follows. Consider the set $S_1$ of all the houses which are ranked first by any agent. If there are two distinct houses $h_1$ and $h_2$ in $S_1$, then the algorithm allocates these two houses to the agents who like them the most. We then order the remaining agents in an arbitrary manner and let each agent choose its highest ranked house among the remaining houses. Notice that the agents who receive $h_1$ and $h_2$ are indeed envy-free, and all the remaining agents are envious. \Cref{claim:atmost2} ensures that this is the best we can hope for. Otherwise, if $|S_1| = 1$ then everyone likes the same house, say $h$ as their first ranked house. Now consider the $span(h)$. If $m-|span(h)| \geq n$, then $span(h)$ remains unallocated. Construct the set $S_{span(h)+1}$ that contains all the houses ranked at $span(h)+1$ by any agent. This set must contain at least two distinct houses, say $h_1$ and $h_2$. Allocate $h_1$ and $h_2$ to agents $i_1$ and $i_2$ such that they rank $h_1$ and $h_2$ respectively at $|span(h)|+1$. This gives us two envy-free agents (since $span(h)$ is unallocated). The remaining agents are then ordered in an arbitrary manner and allocated their best available house from $m - |span(h)|$ houses. Otherwise, if $m-|span(h)| < n$, then we allocate $h$ non-wastefully and allocate the remaining $n-1$ agents their best available house.

We are now ready to prove the main result.

\begin{proof}[Proof of \Cref{thm:single-dipped}] We show that \Cref{alg:numberenvySD} correctly outputs an allocation $A$ with minimum $\numenvy$. If EF$(A) = 2$, then we are done by \Cref{claim:atmost2}. 
Else, EF$(A) = 1$. This implies that $m-span(h)<n$. Suppose, for contradiction, there is a complete allocation $A^*$ such that EF$(A^*) = 2$. Then, none of the houses from $span(h)$ could have been allocated under $A^*$. Indeed, all the agents have identical ranking for the houses in $span(h)$, and allocating any of them creates $n-1$ envious agents. Therefore, $span(h)$ must remain unallocated. But then, we have that $A^*$ is not a complete allocation since $m - |span(h)| < n$. This contradicts our assumption.
\end{proof}

When there are ties at the dip, we can have more than $2$ envy-free agents, but we can still minimize $\numenvy$ in polynomial time.

\begin{lemma} When the preferences are single-dipped with ties at the dip, either all the  $n$ agents are envy-free or at most $2$ agents are envy-free under any complete allocation.
\end{lemma}

\begin{proof} If there are $n$ houses that are all tied and ranked last by all the agents, then arbitrarily allocating these $n$ houses creates $n$ envy-free agents. Else, if there are less than $n$, say $n-1$ houses that are all tied and ranked last by all the agents, then even if all of them are allocated among any $n-1$ agents, all such $n-1$ agents will envy the other agent who gets a house outside of the ties.
Therefore, if there are less than $n$ houses in the ties, then by \Cref{claim:atmost2}, there can be at most $2$ envy-free agents.
\end{proof}

\begin{theorem} Given an instance $\mathcal{I} = (N, H, \succ, \rhd)$ of house allocation with single-dipped preferences with ties at the dip, minimizing the number of envious agents admits a polynomial time algorithm.
\end{theorem}

\begin{proof} 
If there aren't enough houses (at least $n$) in the ties, then we proceed similarly as in the proof of \Cref{thm:single-dipped} to arrive at an allocation with two envy-free agents. This settles our claim.
\end{proof}


\subsection{Fairness and Efficiency}
\label{sec:PO}

In this section, we study the compatibility of Pareto optimality with minimizing $\numenvy$ and show the following result.

\begin{restatable}{theorem}{singlepeakedPO}
\label{thm:singlepeakedPO}
    Given a single-peaked/dipped instance $\mathcal{I} = (N, H, \succ, \rhd)$ of house allocation, we can decide in polynomial time whether there exists an allocation that minimizes $\numenvy$ and is Pareto optimal.
\end{restatable}

The following example illustrates that an envy-free allocation may not be Pareto optimal.
\begin{example}
    Consider the following instance with three agents and four houses. The highlighted allocation $A$ is envy-free since the contentious house $h_1$ remains unallocated. But it is not Pareto optimal since allocating $h_1$ to any of the agents strictly improves their value without making anyone else worse off.

\begin{table}[!ht]
    \centering
    \begin{tabular}{ll}
    $i_1$: & $h_1 \succ \mathbf{h_2} \succ h_3 \succ h_4$ \\
    $i_2$: & $h_1 \succ \mathbf{h_3} \succ h_2  \succ h_4$ \\
    $i_3$: & $h_1 \succ \mathbf{h_4} \succ h_2 \succ h_3$ \\
    \end{tabular}
    \label{tab:example}
\end{table}
    
\end{example}

We have that a Pareto optimal (PO) allocation may not be compatible with minimizing envy in the system. We next show that we can decide in polynomial time whether there exists an allocation that minimizes $\numenvy$ and is PO. Given an allocation $A$, the envy graph $G$ of $A$ is a directed graph with vertices representing the agents and a directed edge from agent $i$ to agent $j$ if $i$ envies $j$ under $A$. Note that if there is a directed cycle in $G$, then the cyclic exchange of bundles does not create any new envious agents but improves the allocation of the agents on the cycle. We first show the characteristics of a PO allocation.

\begin{restatable}{lemma}{POcharacteristics}
\label{lem:PO} Given an instance $\mathcal{I} = (N, H, \succ, \rhd)$ of house allocation, an allocation $A$ is Pareto optimal if and only if (1) there is no directed cycle in the envy graph of $A$, and (2) each house $h$ such that $rank_i(h) < rank_i(A(i))$ are allocated under $A$ for any $i \in [n]$.
\end{restatable}

\begin{proof}
    Suppose $A$ is Pareto optimal. If there is an envy cycle in the envy graph of $A$, then resolving the envy cycle leads to a Pareto improvement. This contradicts that $A$ was Pareto optimal. Further, if there is an unallocated house $h$ such that $rank_i(h) < rank_i(A(i))$, then allocating $h$ to $i$ is a Pareto improvement. This again contradicts that $A$ was Pareto optimal. 
    
    On the other hand, suppose the two properties hold. If $A$ is not Pareto optimal, then there is an agent, say $i_1$, who can be made better off, without making any other agent worse off. Agent $i_1$ can be made better off only by assigning a house $h$ such that $rank_{i_1}(h) < rank_{i_1}(A(i_1))$. Such a house $h$ must have been allocated under $A$, say to agent $i_2$ (by property $2$). Since there are no envy cycles, it must be that $i_2$ prefers $h$ over $A(i_1)$. If $i_2$ gets $A(i_1)$, then it is worse-off. Otherwise, if $i_2$ ends up getting a better-ranked house which is again allocated under $A$, then since there are no envy cycles, at least one agent must become worse off in the process. Therefore, we have that $A$ is Pareto optimal.
\end{proof}

We now present the proof of Theorem \ref{thm:singlepeakedPO}.


\begin{proof}
Let $A$ be the output of \Cref{alg:numberenvy} that minimizes $\numenvy$. 
If $A$ is Pareto optimal, then we are done. Suppose $A$ is not Pareto optimal. Then we claim that there is no allocation for $\mathcal{I}$ which minimizes $\numenvy$ and is Pareto optimal as well. Suppose the number of envy-free agents under $A$ is $k$. Since $A$ is not Pareto optimal, then by \Cref{lem:PO}, either there is a directed cycle $C$ in the envy graph of $A$ or there is a house $h$ such that $rank_i(h) < rank_i(A(i))$ remains unallocated under $A$. 
In the former case, we obtain an allocation $A'$ such that envy graph of $A'$ no longer contains $C$. To that end, we do a cyclic exchange of bundles of agents in $C$. That is, 
given a directed cycle $C := \{i_1, i_2, \ldots i_c\}$, we obtain $A'$ as $A'(i) = A(i)~\forall~i \notin C$ and $A'(i_1) = A(i_2), A'(i_2)=A(i_3), \ldots A'(i_c) = A(i_1)$. Note that since the set of allocated houses remains the same, the number of envious agents does not increase. \\
If there are no envy cycles in the envy graph of $A$ and $A$ is not Pareto optimal, then there is a house $h$ such that $rank_i(h) < rank_i(A(i))$ remains unallocated under $A$. Then $h$ can not be an individual peak for any agent, else it would have been allocated non-wastefully by  \Cref{alg:numberenvy}. Otherwise, suppose $h$ is a shared peak or belongs to $span(h')$ such that $h'$ is a shared peak. Since it remains unallocated, either $span(h)$ or $span(h')$ must have been resolved by \Cref{alg:numberenvy}. This leads to two envy-free agents corresponding to $span(h)$ or $span(h')$ (by \Cref{claim:atmosttwo}). Therefore, we have $p_I+p_S < k$ (by \Cref{lem:numEF}).
In any Pareto optimal allocation, all the peaks must be assigned (by \Cref{lem:PO}), but then, at most $p_I+p_S$ would be envy-free, hence violating minimum $\numenvy$. Therefore, no PO allocation can minimize $\numenvy$ in $\mathcal{I}$. This settles the claim.

We now argue that deciding the existence of $\numenvy$ and PO allocation admits a polynomial time algorithm for single-dipped preferences as well. The argument is analogous to the single-peaked case. Indeed, consider the allocation $A$ as output by \Cref{alg:numberenvySD}. If it is PO, we are done. Otherwise, if there is any directed cycle in the envy graph of $A$, it can be analogously resolved by a cyclic exchange of bundles without making any agent newly envious. Finally, if there is a house $h$ such that $rank_i(h) < rank_i(A(i))$ and $h$ remains unallocated under $A$, then $h \in span(h')$. Also, $h$ remains unallocated implies that $span(h)$ remains unallocated. So, at least $2$ agents are envy-free under $A$. Under any PO allocation $A'$, $h$ must be allocated and consequently, all the agents that do not receive $h$ will be envious of whoever gets $h$. This implies that at most $1$ agent is envy-free under any PO allocation. Hence, none of the PO allocations minimize $\numenvy$. The theorem stands proved. 
\end{proof}

\section{Experiments}


We experimentally evaluated our algorithms with synthetic house allocation instances. 
%
We construct instances with $n = 6$ agents and $6 \leq m \leq 11$ houses. Every agent values a house between an integer between $0$ and $10$ chosen uniformly at random. The results are averaged over 100 instances for each $(n, m)$ pair\footnote{The codes can be accessed at \url{https://github.com/anonymous1203/FairSocieties}}. 
For each instance, our algorithm (\Cref{alg:minEnvylocal}) is initialized with a welfare-maximizing allocation. We compare the welfare of the allocations obtained as the numbers of reallocations aimed at minimizing $\numenvy$ increases $1 \leq q \leq n$. 
%

\begin{figure*}[t]
    \begin{minipage}{0.48\textwidth}
        \includegraphics[scale=0.25]{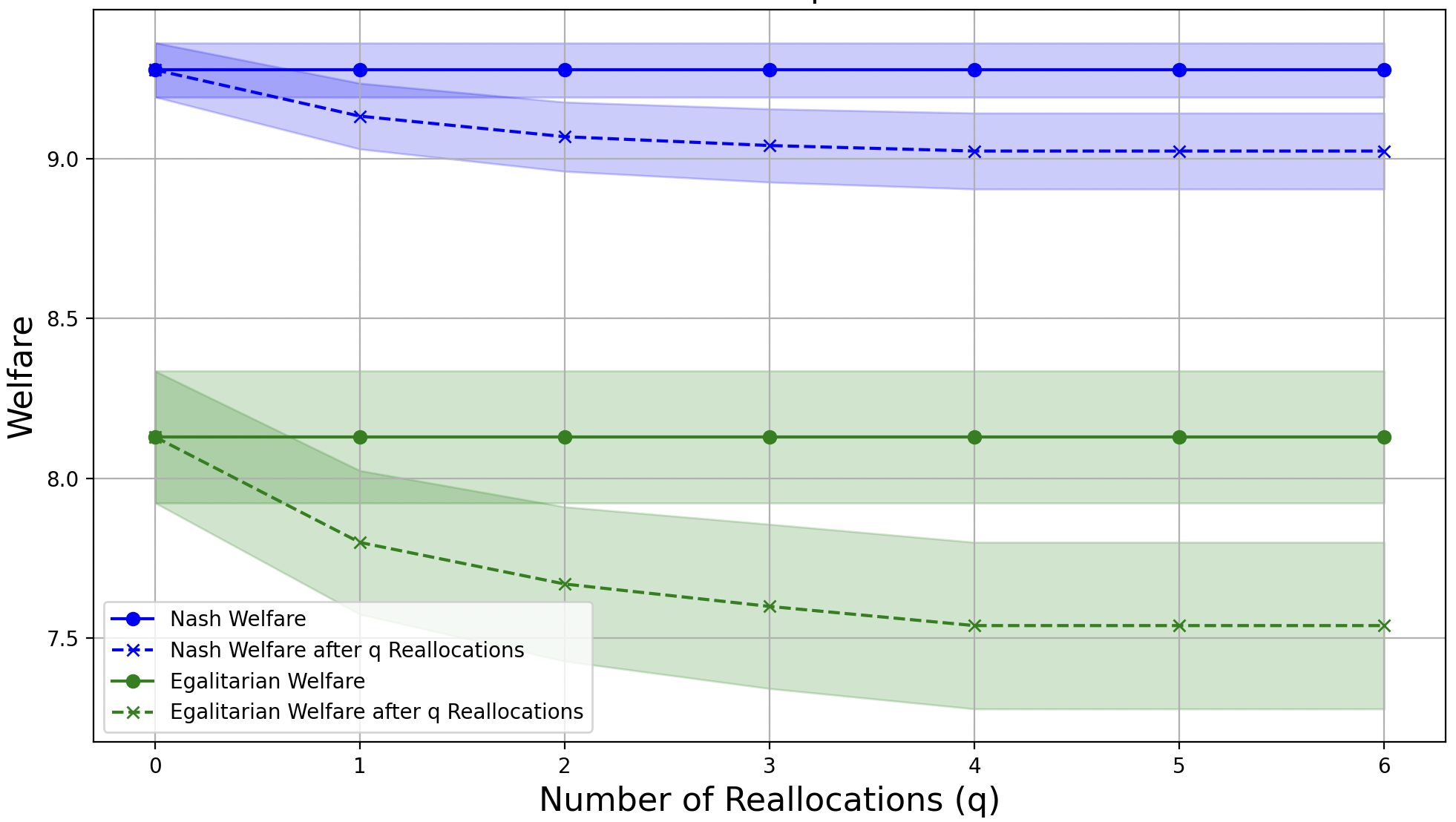}
        \caption{Welfare loss incurred starting from a Nash (blue) and Egalitarian (green) welfare-maximizing allocation and performing at most $q$ reallocations, where $1 \leq q \leq n$.}\label{fig:welfare_nash_and_egal} 
    \end{minipage}\hfill
    \begin{minipage}{0.45\textwidth}
        \includegraphics[scale=0.25]{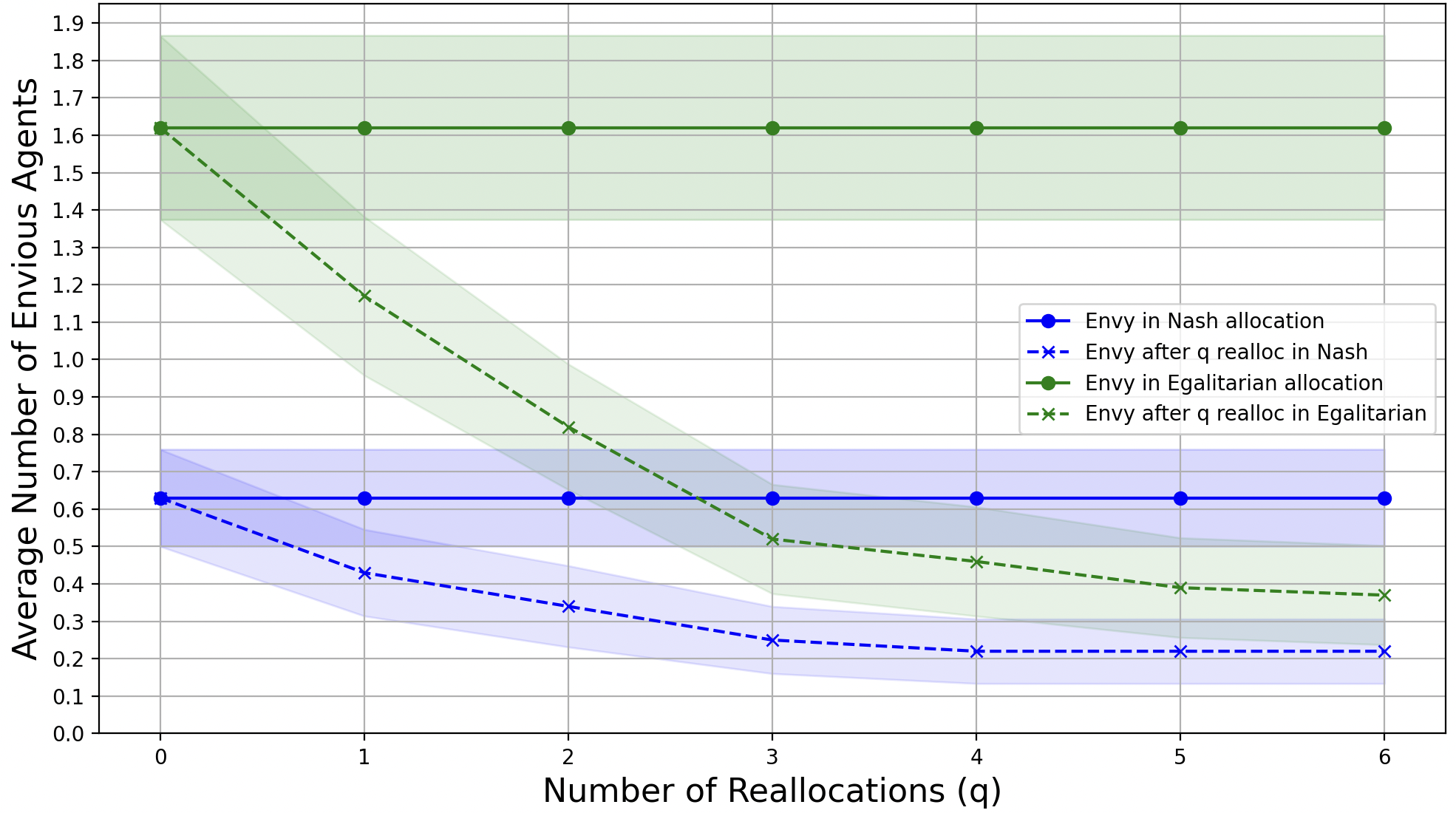}
        \caption{The drop in $\numenvy$ starting from Nash (blue) and Egalitarian (green) welfare-maximizing allocations and performing at most $q$ reallocations, where $1 \leq q \leq n$.}\label{fig:nash_and_egal}
    \end{minipage}
\end{figure*}

The symmetric difference between welfare maximizing allocation and the one that minimizes $\numenvy$ gives us the desired alternating paths and cycles. We then use a dynamic program to select the optimal set of $q$ components that guide the reallocations. 

Figure \ref{fig:welfare_nash_and_egal} shows that the welfare loss as we increase the number of reallocations starting from a Nash and an Egalitarian welfare-maximizing allocation. 
Figure \ref{fig:nash_and_egal} shows the reduction in $\numenvy$ as we increase the number of reallocations. Both plots suggest that welfare loss and the drop in $\numenvy$ starts flattening after $3$ reallocations. That is, a small number of reallocations, specifically around $3$, are sufficient to significantly reduce envy while maintaining high welfare. Thus, starting from a welfare-maximizing allocation, just a few targeted changes can move the system close to a minimum envy state without substantial sacrifice in welfare.
We also visualize how an allocation with a fixed number of reallocations $q$ compares in terms of welfare with the one that minimizes $\numenvy$ and $\tenvy$ and the one that maximizes the welfare (see Figure \ref{fig:fourfigs} and \ref{fig:fourfigs_total_envy} in the appendix). We defer the additional experiments, including those with single-peaked/dipped preferences, to the appendix.

\section{Conclusion} 
We present a general framework that enables tractable computations for finding fairer solutions in house allocations. Given the known hardness and inapproximability of minimizing envy in this context (for instance, hardness of minimizing maximum envy even when the preferences are binary, maximum envy is one, and the maximum degree in the preference graph is constant \citep{Madathil2025}), unless P=NP, we cannot get a tractable algorithm parameterized by maximum envy $k$ and $d$, even for binary preferences, making a framework like ours essential to achieve any tractable solution.  
The properties that are necessary for our algorithm are: the connectivity of minimal improvement sets, and being able to compute the decrease in envy ($r_C$) due to an improvement set. Thus, the algorithm may extend for weighted envy~\citep{dai2024weighted} but would not directly work for concepts like local envy~\citep{HPSVV23,BCGHLMW2019} where the envy depends on factors other than the structure of the preferences. We also show that minimizing $\numenvy$, which is known to be hard for general rankings, becomes tractable on single-peaked/dipped preferences. Extending these algorithms for such rankings with ties and for other measures of envy is an interesting direction.

\section*{Acknowledgements}

HH acknowledges the support from the National Science Foundation (NSF) through CAREER Award IIS-2144413 and Award IIS-2107173. SR is supported by the Start-Up project grant of Indian Statistical Institute.
AS acknowledges the support from Walmart Center for Tech Excellence (CSR WMGT-23-0001).

\bibliographystyle{plainnat}
\bibliography{references.bib}
\clearpage

\clearpage

\appendix

\section{Preliminaries}
\label{sec:cardinal_prelims}
We implement our algorithms on house allocation instances with cardinal preferences. Here, we define such instances and the associated welfare measures.

\paragraph{Cardinal Preferences.}
An instance of the house allocation problem with cardinal preferences is denoted by $\langle N, H, v \rangle$, where $N:= \{1, 2, \ldots n\}$ is the set of $n \in \mathbb{N}$ agents, $H:= \{h_1,h_2,\dots,h_m\}$ is the set of $m \in \mathbb{N}$ houses, and $V:= \{v_{1}, v_{2}, \ldots, v_{n}\}$ is the \emph{valuation profile}. Each $v_i: H \rightarrow \mathbb{Z}$ indicates agent $i$'s value for house $h \in H$. An agent is said to derive a utility of $v_i(A(i)$ under allocation $A$.

\paragraph{Fairness Notions.}
    Given an allocation $A$, we say that agent $i$ \emph{envies} agent $j$ if it values $j's$ house more than its own house. That is, if $v_{i}(A(j)) > v_i (A(i))$. The \textit{amount} of this pairwise envy is defined as $\envy_{i,j} (A) := \max\{v_i (A(j)) - v_i (A(i)), 0\}$ for cardinal preferences.

\paragraph{Welfare Notions.} The welfare of an allocation $A$ is defined as some aggregation of individual agents' utilities  that they derive from their allocated houses under $A$. Specifically, Utilitarian welfare $UW(A)$ of an allocation $A$ is the sum of the individual agents' utilities: $UW(A) = \sum_{i \in N} v_i(A(i)$; Nash welfare ($NW(A)$) is defined as the geometric mean of the individual agents' utilities: $NW(A) = (\prod_{i \in N} v_i(A(i)))^{\frac{1}{n}}$ and Egalitarian welfare ($EW(A)$) is the minimum utility of any agent under $A$: $EW(A) = \min_{i \in N} v_i(A_i)$.

\section{Missing Proofs from \Cref{sec:fpt}}




\Sisfeasible*

\begin{proof}[Proof of \Cref{thm:fpt_enviousagents}~(Continued)]
 We now argue that there is a solution to the knapsack problem if and only if there is a desired allocation $A$.  
Let the set $\mathcal{C}' \subseteq \mathcal{C}$ be solution to the knapsack problem. 
Using condition (1)-(3) of feasibility, we have that each component in $\mathcal{C}'$ contains one or (possibly) more $\hat{A}$-alternating path(s)/cycle(s).
Let $X$ denotes the set of $\hat{A}$-alternating paths and cycles in $\mathcal{C}'$.  
Thus, we convert the allocation $\hat{A}$ to $ A := \hat{A} \oplus X$. Observe that each feasible component is a nice set by condition (4). Thus, we get that $\numenvy(A) = \numenvy(\hat{A})- \sum_{C \in \mathcal{C}'}r_C$. Therefore, $\numenvy(A) = \numenvy(\hat{A})-k$. Additionally,  we get that $|A~\Delta~\hat{A}| \leq q$ since we have $\sum_{C \in \mathcal{C}'} n_C \leq q$. Hence, we have a yes-instance of the house allocation problem.

On the other hand, suppose we have a yes-instance of the house allocation problem. That is, we have an allocation $A$ such that $\numenvy(A) = \numenvy(\hat{A})-k$ and $|A~\Delta~\hat{A}| \leq q$. Then, consider the symmetric difference $A~\Delta~\hat{A}$, which contains $\hat{A}$-alternating paths and cycles $T_1, \ldots, T_{s'}$  in $G$. 
Moreover, there exists a collection of nice sets $S_1, S_2, \dots S_\ell$ such that $\{T_1, \ldots, T_{s'}\} = S_1~ \dot\cup ~S_2~ \dot\cup ~\dots~ \dot\cup ~S_\ell$. Let $C_i$ denote the component induced on the vertices of $S_i$ in $G$, that is, $C_i= G[V(S_i)]$ for $i \in [\ell]$.
Note that in $\hat{A}$-alternating path/cycle $T_i$ for $i \in [s']$, we have that each agent, house, and edge is colored red by $\chi$ since $\chi$ is a good coloring. Moreover, by \Cref{lem:Sisfeasible}, we have that $C_i$ is a feasible component in $\mathcal{C}$ for each $i \in [\ell]$. 
 Then, we have $\sum_{i \in [\ell]} n_{C_i} \leq q$ and $\sum_{i \in[\ell]}r_{C_i} \geq k$ and hence, it is a yes-instance of the knapsack problem.
This completes the proof for the equivalence.
\end{proof}

 \subsection{Derandomization of \Cref{alg:minEnvylocal}}
 \label{sec:derandomization}
 To make \Cref{alg:minEnvylocal} deterministic, we first introduce the notion of an $n$-$p$-$q$-{\em lopsided universal} family. Given a universe $U$ and an integer $\ell$, we denote all the $\ell$-sized subsets of $U$ by ${U \choose \ell}$. We say that a  family 
$\mathcal{F}$ of sets over a universe $U$ with $\vert U\vert=n$, is an $n$-$p$-$q$-{\em lopsided universal} family if for every $A \in {U \choose p}$ and $B \in {U \setminus A \choose q}$, there is an $F \in \mathcal{F}$ such that $A \subseteq F$ and $B \cap F = \emptyset$.

\begin{lemma}[\cite{FominLPS16}]
\label{lem:lopsidedUniversal}
There is an algorithm that given $n,p,q\in {\mathbb N}$ constructs an  $n$-$p$-$q$-lopsided universal family $\mathcal{F}$ of cardinality ${p+q \choose p} \cdot 2^{o(p+q)}  \log n$ in time 
$\vert  \mathcal{F} \vert  n$. 
\end{lemma}

To de-randomize the algorithm, we replace the function $\chi$ in our algorithm, we will use two steps. Let $\hat{n}=n+m$ and $\hat{m}$ denote the number of edges in the preference graph $G$. First, we use an $\hat{n}+\hat{m}$-$3q+qd$-$2qd$-lopsided universal family $\mathcal{F}_1$  over the vertices and edges of the preference graph $G$, where $A$ will be the red and green vertices and edges and $B$ will be the blue edges as defined in good coloring. Then,  family $\mathcal{F}_1$ has cardinality ${3q(1+d) \choose 2qd} \cdot 2^{o(dq)}  \log (\hat{n}+\hat{m})$. Then, to separate the red and green edges, we use  $\hat{m}$-$q$-$qd$-lopsided universal family $\mathcal{F}_2$ of cardinality ${q+qd \choose q} \cdot 2^{o(dq)}  \log \hat{m}$, where $\mathcal{F}_2$ is a family over the edge set of $G$. For every set $F_1\in \mathcal{F}_1$, we create a function $f_1$ that colors the vertices and edges of $F$ as $1$, and colors all the other vertices and edges as $2$. Similarly, for every set $F_2\in \mathcal{F}_2$, we create a function $f_2$ that colors every edge of $F_2$ that is colored $1$ by $f_1$ as $1$, and colors all the other edges that is colored $1$ by $f_1$ as $2$. Now, for every pair of functions $(f_1,f_2)$, where $F_1\in \mathcal{F}_1$ and $F_2\in \mathcal{F}_2$, we run our algorithm described above. 
\Cref{lem:lopsidedUniversal} implies we will find a pair of functions $(f_1,f_2)$ corresponding to a good coloring. Hence, \Cref{thm:fpt_enviousagents} can be proved using the two universal families.

\begin{proof}[Proof of Theorem \ref{thm:tenvy_maxenvy}]
We argue that with suitable modification to the definition of $r_C$ in  \Cref{alg:minEnvylocal}, it outputs an allocation $A$ such that the total envy reduces by $k$, that is, $\tenvy(A) = \tenvy(\hat{A})-k$ and $|A~\Delta~\hat{A}| \leq q$. To that end, in the Step $8$ of \Cref{alg:minEnvylocal}, we define $r_C := \tenvy(\hat{A}) - \tenvy(\hat{A} \oplus C)$. Dependent subsets and improvement sets are defined analogously for total envy. Then, the results (Lemmas \ref{lem:S_connected}-\ref{lem:Sisfeasible}) hold in this context envy as well as they do not depend on the definition of $r_C$. Similarly, we get an algorithm that reallocates at most $q$ agents and minimizes $\maxenvy$ by $k$ by defining $r_C := \max_{i\in N} \envy_i(\hat{A}) - \max_{i\in N} \envy_i(\hat{A}\oplus C)$ in the Step $8$ of \Cref{alg:minEnvylocal}.
\end{proof}

\begin{remark}\label{rem:localsearchtime}
     We can further optimize our running time to construct a minimal improvement set efficiently using the randomized algorithm (using the good coloring). We can sample the houses to be reallocated to construct a minimal improvement set. At some point we have picked $t$ houses for some $t \in [m]$ and $(q(d+1)-t)/(m-t)$ drops below $1/9$. That is $t$ is the largest value such that $(q(d+1)-t)/(m-t)>1/9$. Then, we run \Cref{alg:minEnvylocal}, spending time $9^{q(d+1)}n^{O(1)}$.
     Therefore, our algorithm have picked a
random set $X$ of size $t$. Let $N_X$ be the set of agents allocated to $X$. We run \Cref{alg:minEnvylocal}, where we start with a partially colored graph $G$ such that $X \cup N_X$ and the set $Y$, neighbors of $X\cup N_X$, in $G$ are colored correctly, assuming $X \cup N_X \subseteq S$ (see the definiton of good coloring). Thus, in Line 2 of \Cref{alg:minEnvylocal}, the coloring function $\chi$ colors vertices $V(G) \setminus (X\cup N_X \cup Y)$ and edges that are not incident to $X \cup N_X$. The next steps remain unchanged. Thus, given the partial coloring, the algorithm takes time $9^{q(d+1) - 2td}n^{O(1)}$.

The set $X$ is a subset of a minimal improvement set with probability $pr = {q(d+1) \choose t}/{m \choose t}$. In order to get constant success probability,
we run the algorithm $1/pr$ times, taking time $9^{q(d+1) - t}n^{O(1)}$ for each run. We optimize the value of $t$ to minimize the run time.
\end{remark}

\section{Missing Proof from \Cref{sec:restricted}}



\singlepeaks*

\begin{proof} We will show the correctness of \Cref{alg:numberenvy}. Let $A$ be the output of \Cref{alg:numberenvy} and let $EF(A)$ denote the set of envy-free agents in the allocation $A$. Suppose $A^*$
 is the allocation that minimizes $\numenvy$. Clearly, $|EF(A^*)| \geq |EF(A)|$. The aim is to show that $|EF(A^*)| = |EF(A)|$. To this end, we will show that $|EF(A^*) \setminus EF(A)| = |EF(A) \setminus EF(A^*)|$.

 Suppose $i \in EF(A^*) \setminus EF(A)$. We will show that corresponding to $i$, there is a unique agent $i^\prime$ such that $i^\prime \in EF(A) \setminus EF(A^*)$. Note that the first-ranked house by $i$, say $h$ must be a shared peak, else $A(i) = h$ and $i \in EF(A)$, which is not the case. Consider the following cases: 
\begin{enumerate}
    \item $A^*(i) = h$. Since $i \notin EF(A)$, there is an agent $i'$ who gets a house $h'$ under $A$ such that $A(i') = h' \succ_i A(i)$.

    \item[1(a).] If $h' = h$, then, since $A(i') = h$, we have that the first ranked house of $i'$ is also $h$. (Since $A$ always allocates the peaks non-wastefully and $h$ is a peak for $a$). This implies that if $i$ is envy-free under $A^*$, then $i'$ must have been envious of $i$ under $A^*$. But, as $A(i') = h$, $i^\prime$ is envy-free under $A$. So, $i^\prime \in EF(A) \setminus EF(A^*)$.

        \item[1(b).] Else, if $h' \neq h$, and $h$ remains unallocated under $A$, then $h$ is definitely a resolved peak under $A$. Suppose there are $k$ overlapping spans with $span(h)$. Then, we must have exactly $k+1$ agents, who are envy-free corresponding to the $k$ overlapping spans (which is the only case when a shared peak is resolved).
        Note that $A^*$ also can have at most $k+1$ envy-free agents corresponding to the above overlapping spans (by $\Cref{lem:atmostk+1}$). Since $i \in EF(A^*) \setminus EF(A)$, we must have some agent $i^\prime$ among the above $k+1$ agents under $A$ such that $i^\prime \in EF(A) \setminus EF(A^*)$.  

  \item[2.] $A^*(i) \neq h$. This implies that $h$ remains unallocated under $A^*$. All houses that $i$ ranks better that $A^*(i)$ remain unallocated (as $i \in EF(A^*)$). Consider all the overlapping spans with span($h$). If there are $k$ of them, then there are at most $k+1$ envy-free agents corresponding to these spans under $A^*$.
    Since $i \notin EF(A)$, there is an agent $i'$ who gets a house $h'$ under $A$ such that $A(i') = h' >_i A(i)$.

    \item[2(a).] If $h' = h$, then the peak $h$ is allocated and hence, not resolved under $A$. We consider the various cases why $h$ was not resolved under $A$ and for each of them, find a corresponding agent who is envy-free under $A$ but is envious under $A^*$. 

Suppose $h$ was not resolved because it is a house that is considered in Case $2$. This implies that $h$ lies in the span($h_j$) for some shared peak $h_j$ and $h_j$ lies in the span of $h$. Then, there can be at most $2$ envy-free agents corresponding to the span($h$) and span($h_j$) in any allocation. Under $A$, there are exactly $2$ envy-free agents, say $\{i_1, i_2\}$, corresponding to the spans since both $h_j$ and $h$ are allocated non-wastefully to say, $i_1$ and $i_2$ respectively. Since $h_j$ is in the span($h$) and $h$ is in the span($h_j$), it is easy to see that under  $A^*$, at most $2$ of the three agents $\{i, i_1, i_2\}$ can be envy-free. Therefore, this implies that $i_2 \in EF(A) \setminus EF(A^*)$ and we are done for this case.

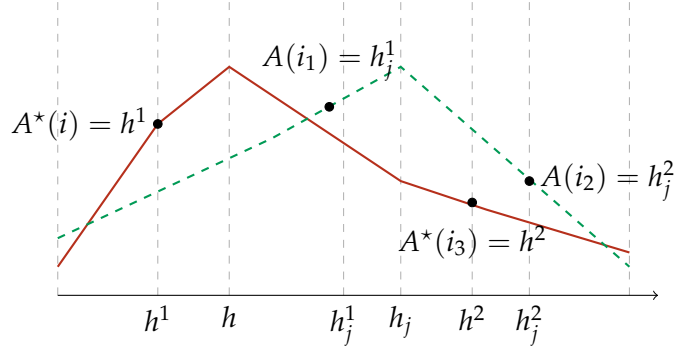
\begin{figure}[htbp]
\centering
\begin{tikzpicture}[scale=1.9]
\def\xlist{0, 0.7, 1.2, 2.0, 2.4, 2.9, 3.3, 4.0}
  \foreach \x in \xlist {
    \draw[gray!60, dashed] (\x,0) -- (\x,2.1);
  }

  \draw[->] (0,0) -- (4.2,0) node[right] {};

  \draw[BrickRed, thick]   (0,0.2) -- (0.7,1.2) -- (1.2,1.6) -- (2.4,0.8) -- (3.0,0.6) -- (4,0.3);
   \draw[ForestGreen, thick, dashed]   (0,0.4) -- (1.5,1.1) -- (2.4,1.6) -- (4,0.2);
  
  \filldraw (0.7,1.2) circle (0.03) node[left] {$A^\star(i)=h^1$};
\filldraw (2.9,0.65) circle (0.03) node[below=5pt] {$A^\star(i_3)=h^2$};
\filldraw (3.3,0.8) circle (0.03) node[right] {$A(i_2)=h^2_j$};
  \filldraw (1.9,1.32) circle (0.03) node[above=6pt] {$A(i_1)=h^1_j$};

  \node[below] at (0,0) {};
  \node[below] at (0.7,0) {$h^1$};
  \node[below] at (1.2,0) {$h$};
  \node[below] at (2,0) {$h^1_j$};
  \node[below] at (2.4,0) {$h_j$};
\node[below] at (2.9,0) {$h^2$};
\node[below] at (3.3,0) {$h^2_j$};

\end{tikzpicture}
\caption{A schematic of Case 2(a) in \Cref{thm:singlepeaks}}
\label{fig:thm5}
\end{figure}
    
Now suppose $h$ was not resolved because it is considered in Case $3$. This implies that there is a shared peak $h_j$ which lies in the span($h$) (so, resolving $h$ forces $h_j$ to remain unallocated). Now by \Cref{lem:atmostk+1}, we know that there can be at most $3$ envy-free agents corresponding to the span($h$) and span($h_j$) in any allocation. In particular, two agents, say $\{i_1, i_2\}$ could be envy-free if span($h_j$) was resolved and $i_3$ is envy-free who is the recipient of $h$ (which is allocated non-wastefully). If $h_j$ was resolved under $A$, then $A$ has exactly $3$ envy-free agents, namely $\{i_1, i_2, i_3\}$ (see \Cref{fig:thm5}). Now, at most $3$ agents from $\{i, i_1, i_2, i_3\}$ can be envy-free under $A^*$. But, as span($h$) is resolved under $A^*$ and $h_j \in span(h)$, so $h_j$ would not have been resolved and is also unallocated under $A^*$. This implies that only two agents, namely $i$ and $i_3$ are envy-free under $A^*$. If $h_j$ is resolved under $A$, then we have three agents $\{i_1, i_2, i_3\}$ envy-free under $A$, contradicting the optimality of $A^*$. And, if $h_j$ is not resolved under $A$, then we have two agents envy-free under $A$, namely $i_1$ (recipient of $h_j$) and $i_3$ (recipient of $h$) (WLOG).
Therefore, we get an agent $i_2$, such that $i_2 \in EF(A) \setminus EF(A^*)$.

    Otherwise, the only reason that $h$ was not resolved under $A$ was the fact that the number of unallocated houses at this point minus span($h$) would have been strictly less than the number of unallocated agents, say $n'$. Since $A$ resolves the spans in the increasing order of their sizes, let $\{h_{r_1}, h_{r_2}, \ldots h_{r_i}\}$ be the set of resolved spans and $\{h_{r_{i+1}} \ldots h \ldots h_{r_t}\}$ be the set of unresolved spans under $A$, as considered in Case $4$.
    Suppose $m'$ and $n'$ are the remaining houses and agents at the beginning of Case $4$. Note that we have $m' - span(h_j) - \sum_{i \in [r_i]} span(h_i)  < n'$ for all $j \in [r_{i+1}, r_t]$.  In particular, $m' - span(h) - \sum_{i \in [r_i]} span(h_i)  < n'$.

   We now argue that if span$(h)$ was resolved under $A^*$, then there must exist at least one span in the set of resolved spans under $A$, $\{h_{r_1}, h_{r_2}, \ldots h_{r_i}\}$, which is not resolved under $A^*$. If not, and all the resolved spans under $A$ are also resolved under $A^*$, then it must be that $m' - span(h) -  \sum_{i \in [r_i]} span(h_i)  \geq n'$. This contradicts the fact that $m' - span(h_{r_{i+1}}) - \sum_{i \in [r_i]} span(h_i) < n'$ since $span(h_{r_{i+1}}) < span(h)$ according to the ordering in Case $4$. Therefore, we have a span which is resolved in $A$ (say, $i_1$ and $i_2$ are two corresponding envy-free agents) but not in $A^*$. Therefore, at most one of $\{i_1, i_2\}$ is envy-free under $A^*$ and WLOG, we have $i_2 \in EF(A) \setminus EF(A^*)$.

 \item[2(b).]  Else, if $h' \neq h$, then the peak $h$ remains unallocated under $A$. Then $h$ is definitely a resolved peak. There are two envy-free agents $i_1$ and $i_2$ under $A$ who receive the resolved peaks and span($h$) remains unallocated. Now consider the set of agents $\{i, i_1, i_2\}$ under $A^*$. Since $i \in EF(A^*)$, at most one of $i_1$ and $i_2$, say $i_1$, can be envy-free under $A^*$ (by \Cref{claim:atmosttwo}). Therefore, $i_2 \in EF(A) \setminus EF(A^*)$.
 \end{enumerate}
 This settles the claim.
\end{proof}

\section{Additional Experiments}

\begin{figure*}[t]
    \centering
    \begin{minipage}[t]{0.48\textwidth}
        \centering
        \includegraphics[scale=0.25]{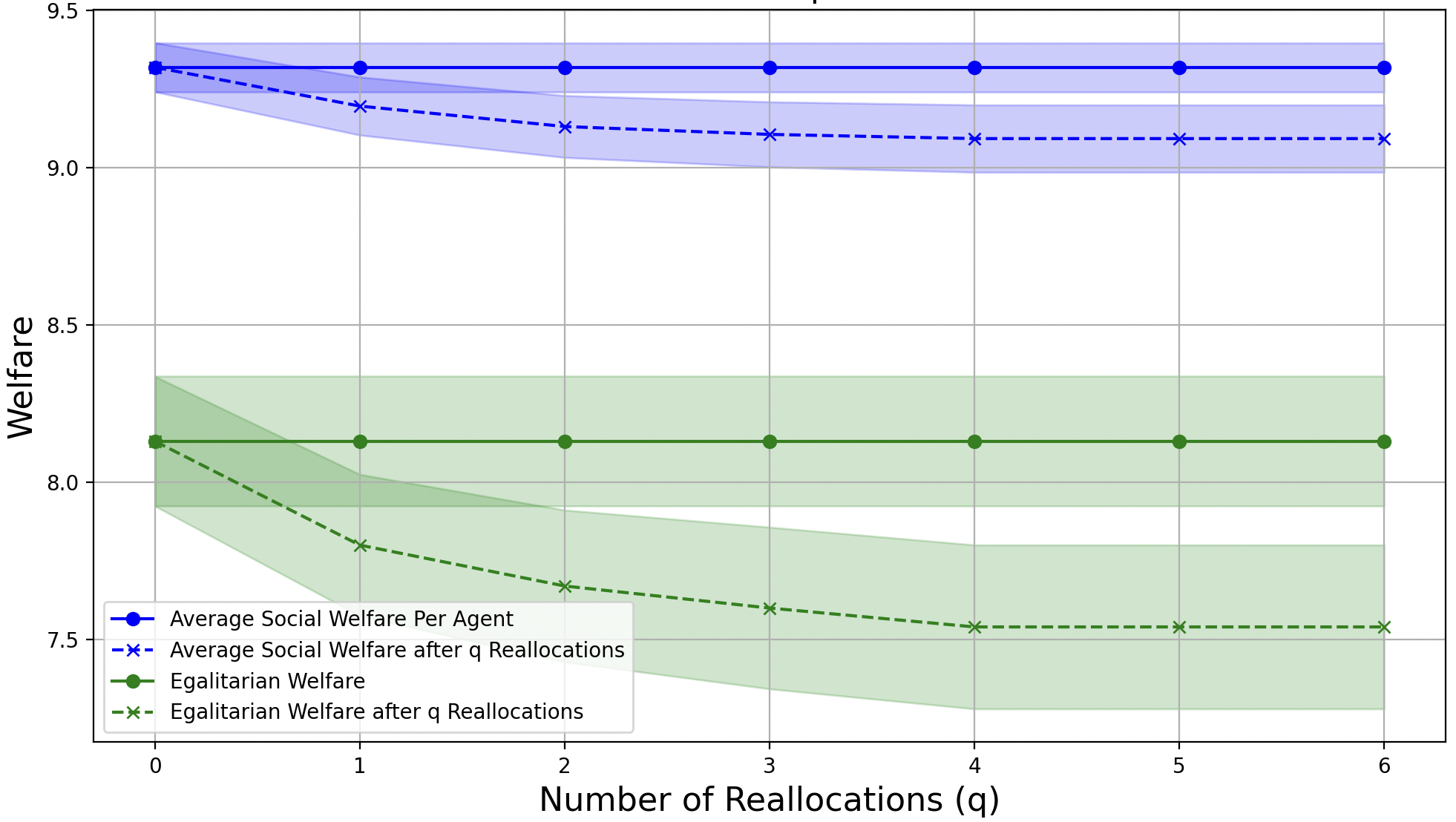}
        \caption{Welfare loss incurred starting from a Utilitarian (blue) and Egalitarian (green) welfare-maximizing allocation and performing at most $q$ reallocations, where $1 \leq q \leq n$.}
        \label{fig:welfare-loss}
    \end{minipage}%
    \hfill
    \begin{minipage}[t]{0.48\textwidth}
        \centering
        \includegraphics[scale=0.25]{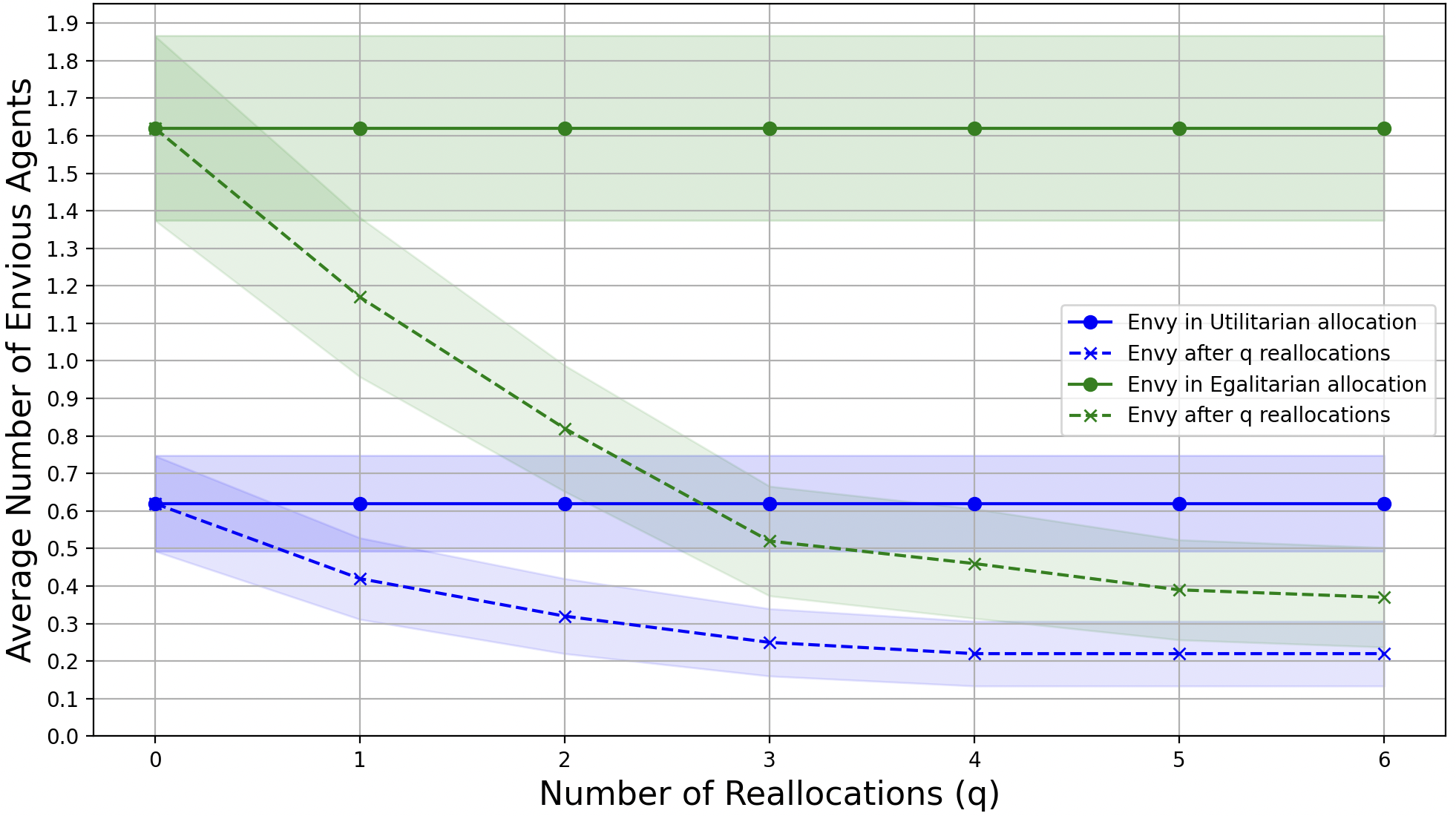}
        \caption{The drop in the number of envious agents starting from Utilitarian (blue) and Egalitarian (green) welfare-maximizing allocation and performing at most $q$ reallocations, where $1 \leq q \leq n$.}
        \label{fig:envy-drop}
    \end{minipage}
\end{figure*}

\paragraph{Bounded number of reallocations.} 
Figure \ref{fig:welfare-loss} depicts the welfare loss as we increase the number of reallocations, starting from a Utilitarian and Egalitarian welfare-maximizing allocation, while Figure \ref{fig:envy-drop} depicts the drop in $\numenvy$. 
We also visualize how an allocation with a fixed number of reallocations $q$ compares in terms of welfare with the one that minimizes $\numenvy$ and $\tenvy$ and the one that maximizes the welfare (see Figure \ref{fig:fourfigs} and \ref{fig:fourfigs_total_envy}).



\paragraph{Single-Peaked / Single-Dipped Preferences.}

We compare the welfare of the Utilitarian welfare-maximizing allocation and that of the allocation that minimizes the number of envious agents.
We ran Algorithm \ref{alg:numberenvy} and \ref{alg:numberenvySD} averaged over $1000$ house allocation instances with $n = 10$ and $10 \leq  m \leq 25$. We also computed welfare-maximizing allocations for each of these instances. Figure \ref{fig:singlepeaked} and \ref{fig:singledipped} depict the average welfare loss in the instances with cardinal preferences conforming to a single-peaked and single-dipped structure. Although the welfare loss is not substantial to begin with, as the number of houses increases, the optimal fair allocation minimizing $\numenvy$ starts achieving the best possible welfare.


\begin{figure*}[ht]
    \centering
    \begin{minipage}{0.48\textwidth}
        \centering
        \includegraphics[scale=0.3]{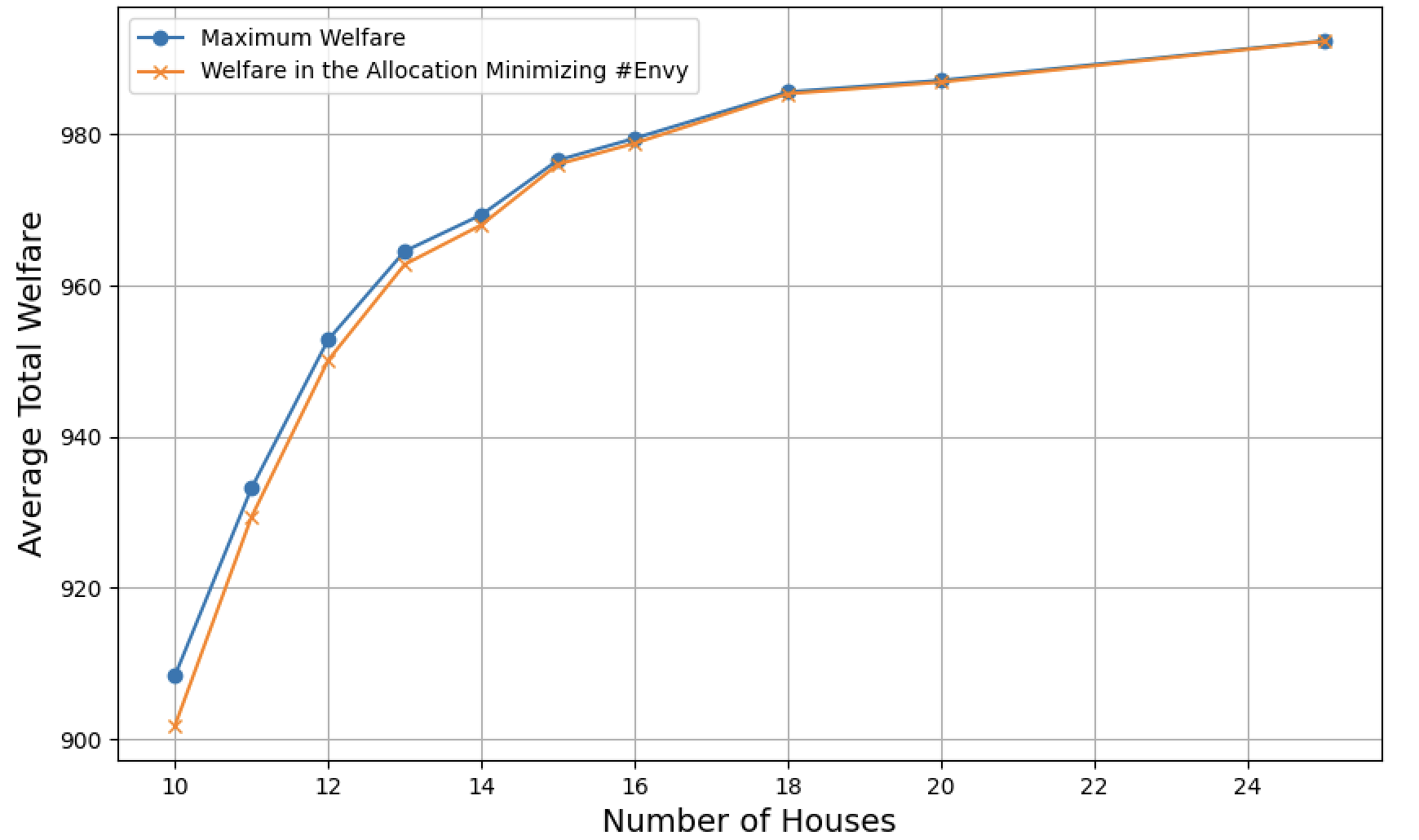}
        \caption{Welfare loss incurred by the allocation that minimizes $\numenvy$ on Single-Peaked Preferences.}
        \label{fig:singlepeaked}
    \end{minipage}\hfill
    \begin{minipage}{0.48\textwidth}
        \centering
        \includegraphics[scale=0.3]{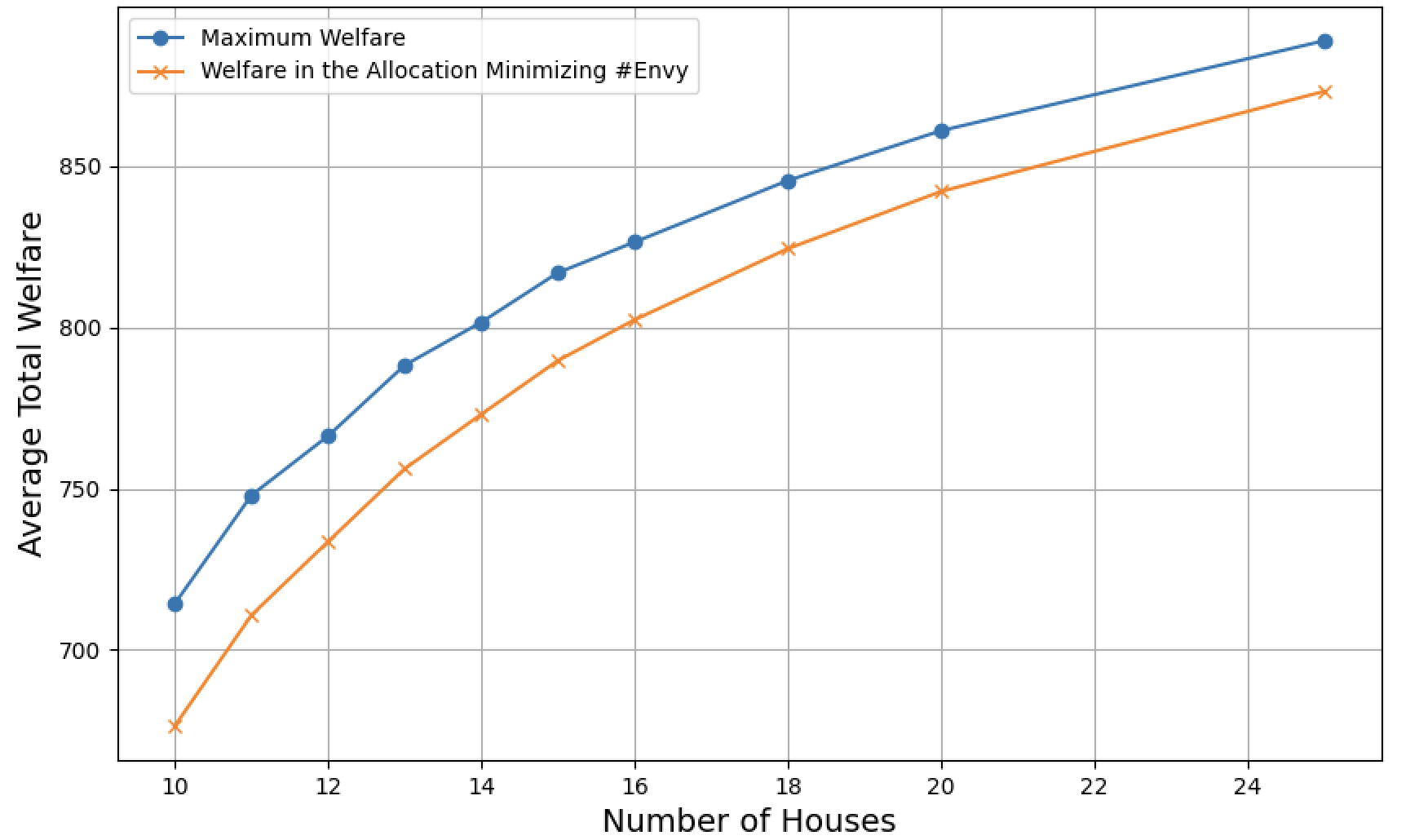}
        \caption{Welfare loss incurred by the allocation that minimizes $\numenvy$ on Single-Dipped Preferences.}
        \label{fig:singledipped}
    \end{minipage}
\end{figure*}

\begin{figure*}[h]
  \centering
  \begin{minipage}[t]{0.45\textwidth}
    \centering
    \includegraphics[width=\linewidth]{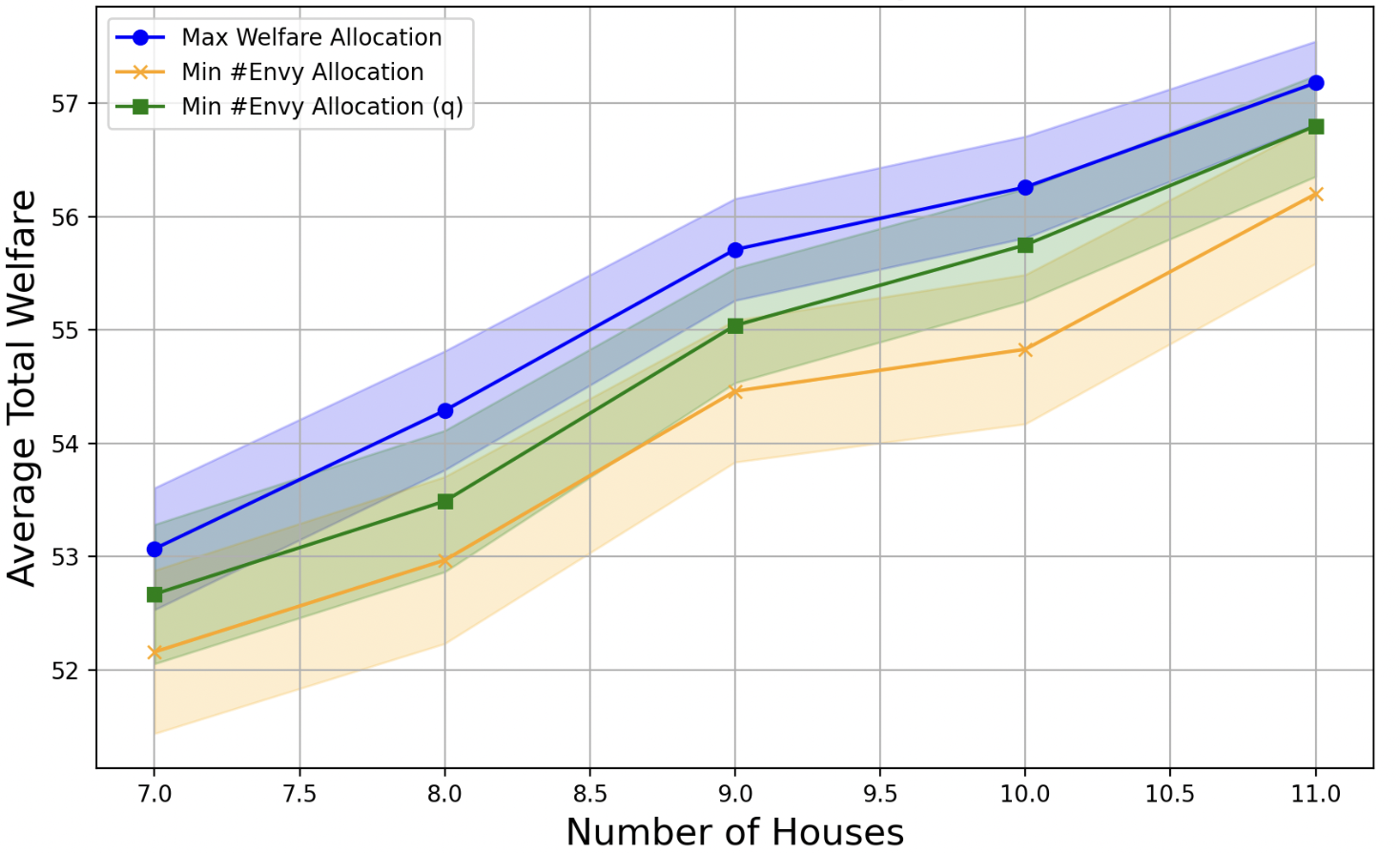}
    \caption*{$q= 1$}
  \end{minipage}
  \hfill
  \begin{minipage}[t]{0.45\textwidth}
    \centering
    \includegraphics[width=\linewidth]{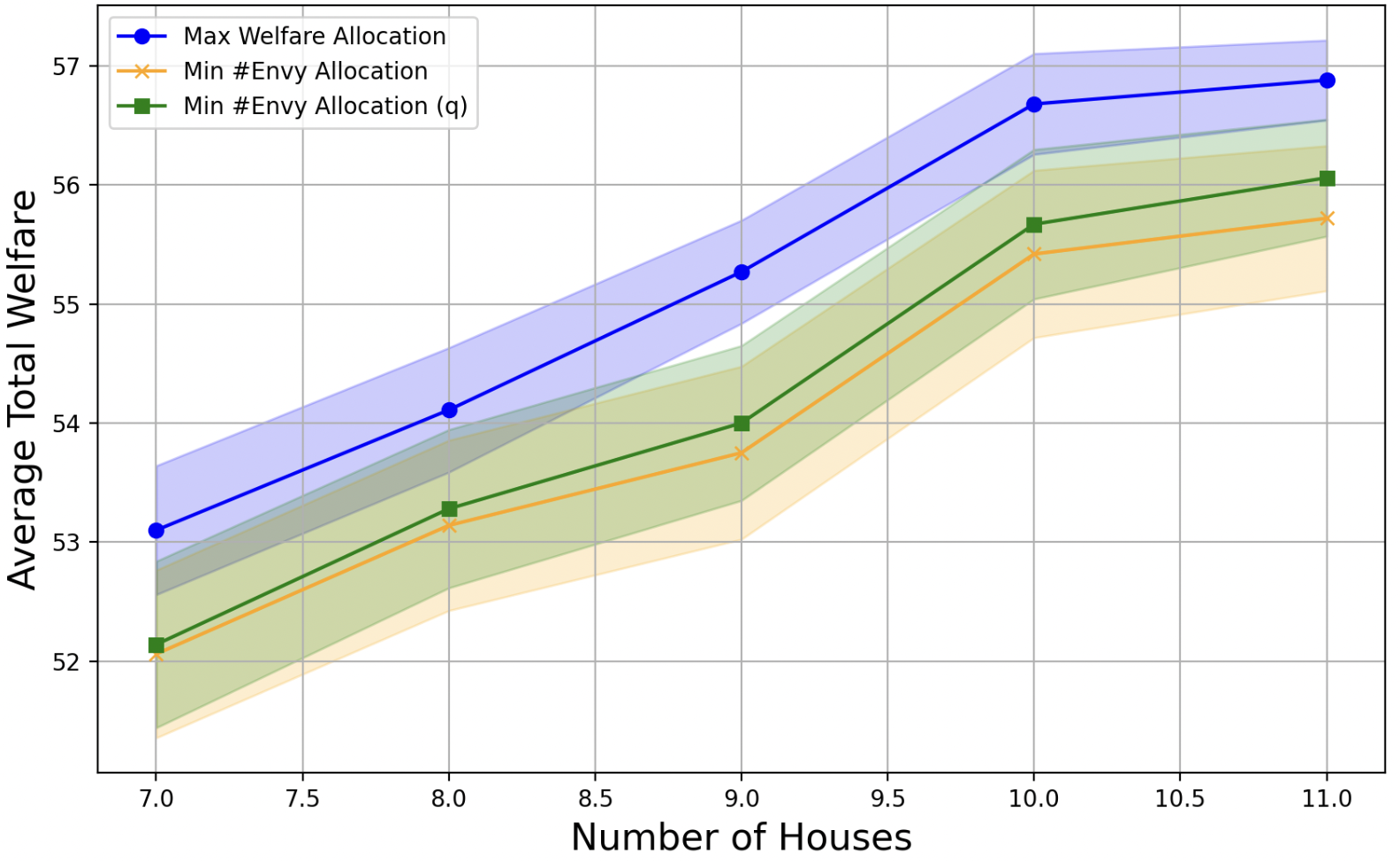}
    \caption*{$q= 3$}
  \end{minipage}
  
  \vspace{0.5cm}
  
  \begin{minipage}[t]{0.45\textwidth}
    \centering
    \includegraphics[width=\linewidth]{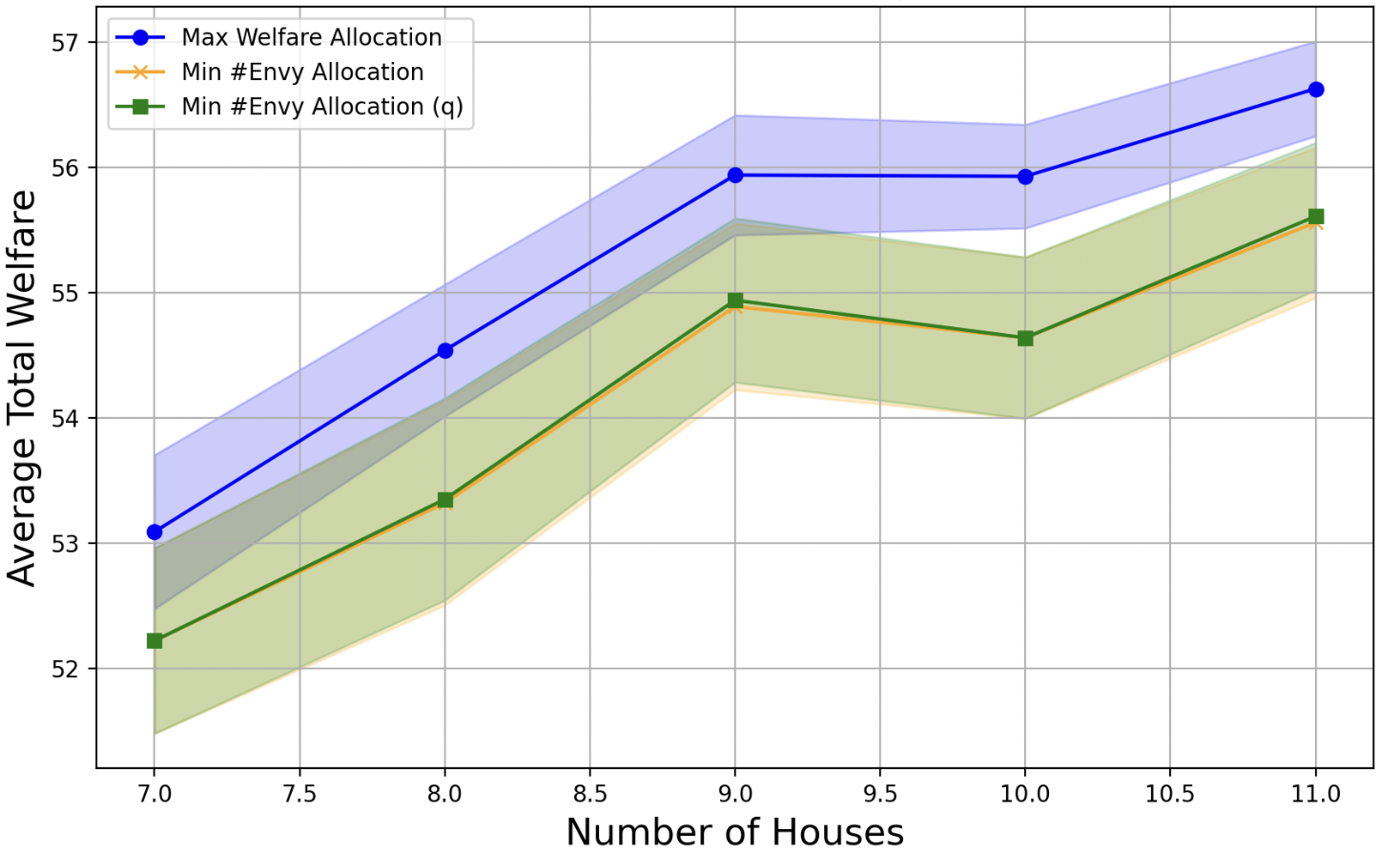}
    \caption*{$q=4$}
  \end{minipage}
  \hfill
  \begin{minipage}[t]{0.45\textwidth}
    \centering
    \includegraphics[width=\linewidth]{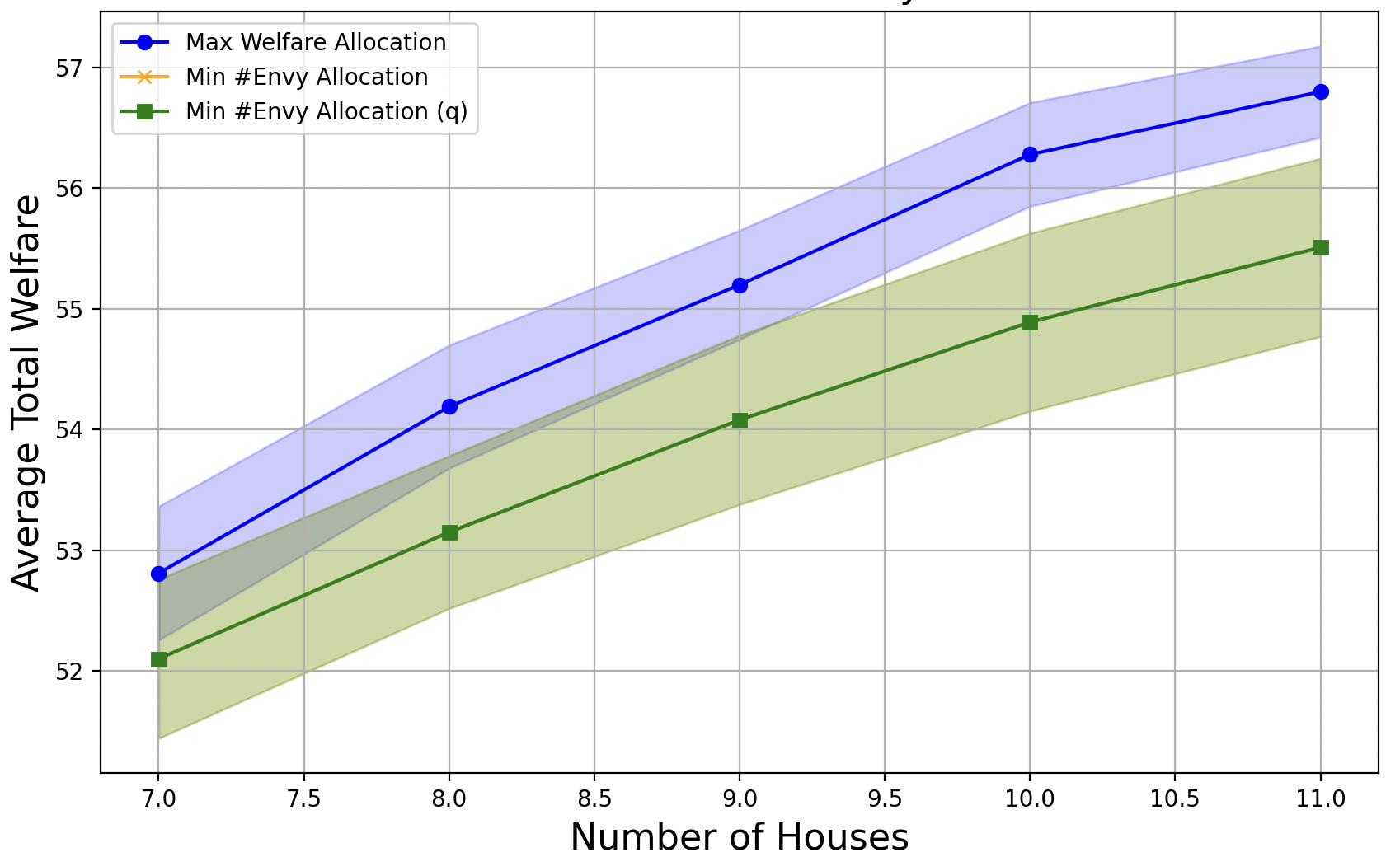}
    \caption*{$q= 6$}
  \end{minipage}
  
  \caption{Welfare loss incurred in allocations that minimize $\numenvy$ with at most $q$ reallocations.}
  \label{fig:fourfigs}
\end{figure*}

\begin{figure*}[h]
  \centering
  \begin{minipage}[t]{0.45\textwidth}
    \centering
    \includegraphics[width=\linewidth]{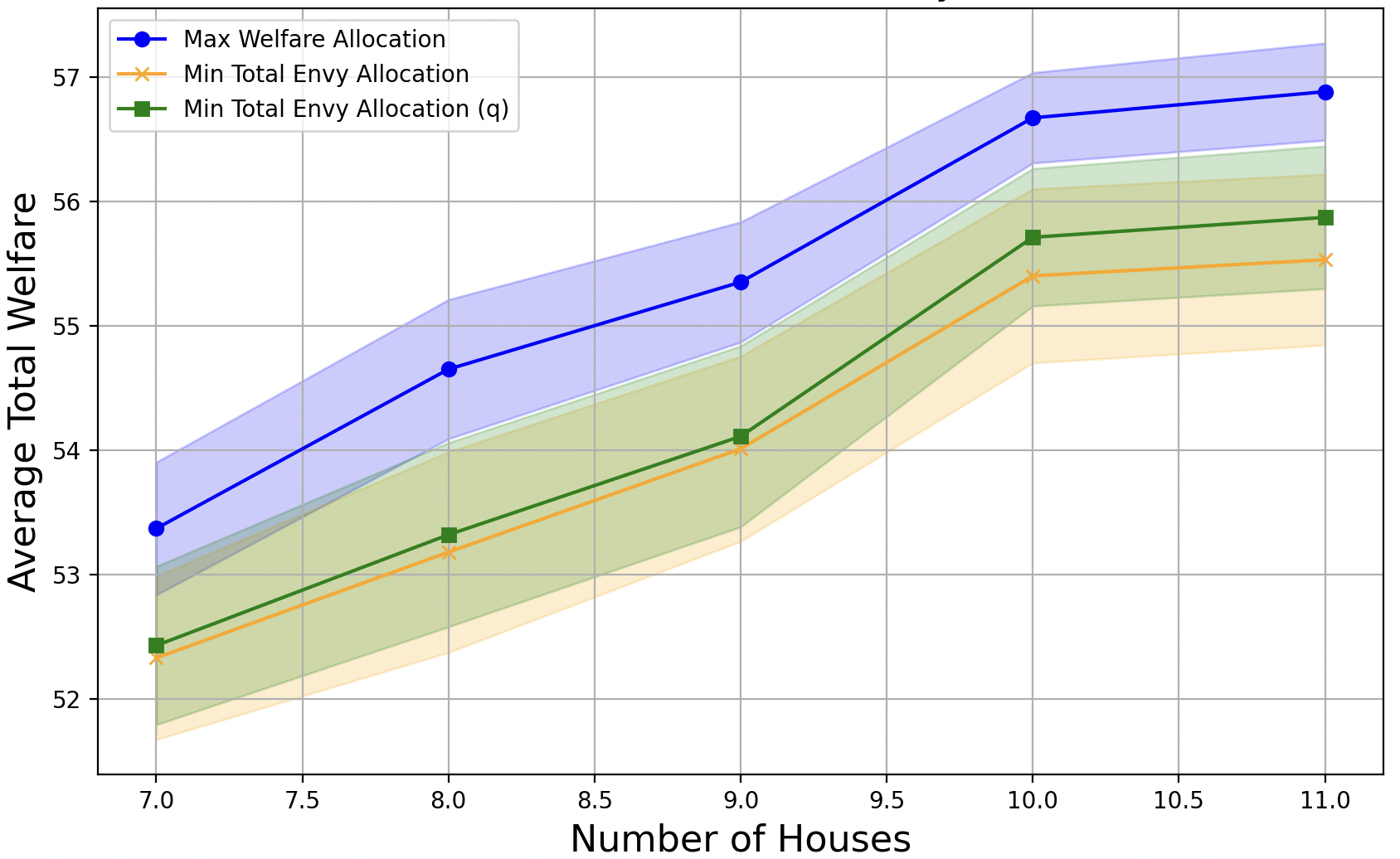}
    \caption*{$q= 1$}
  \end{minipage}
  \hfill
  \begin{minipage}[t]{0.45\textwidth}
    \centering
    \includegraphics[width=\linewidth]{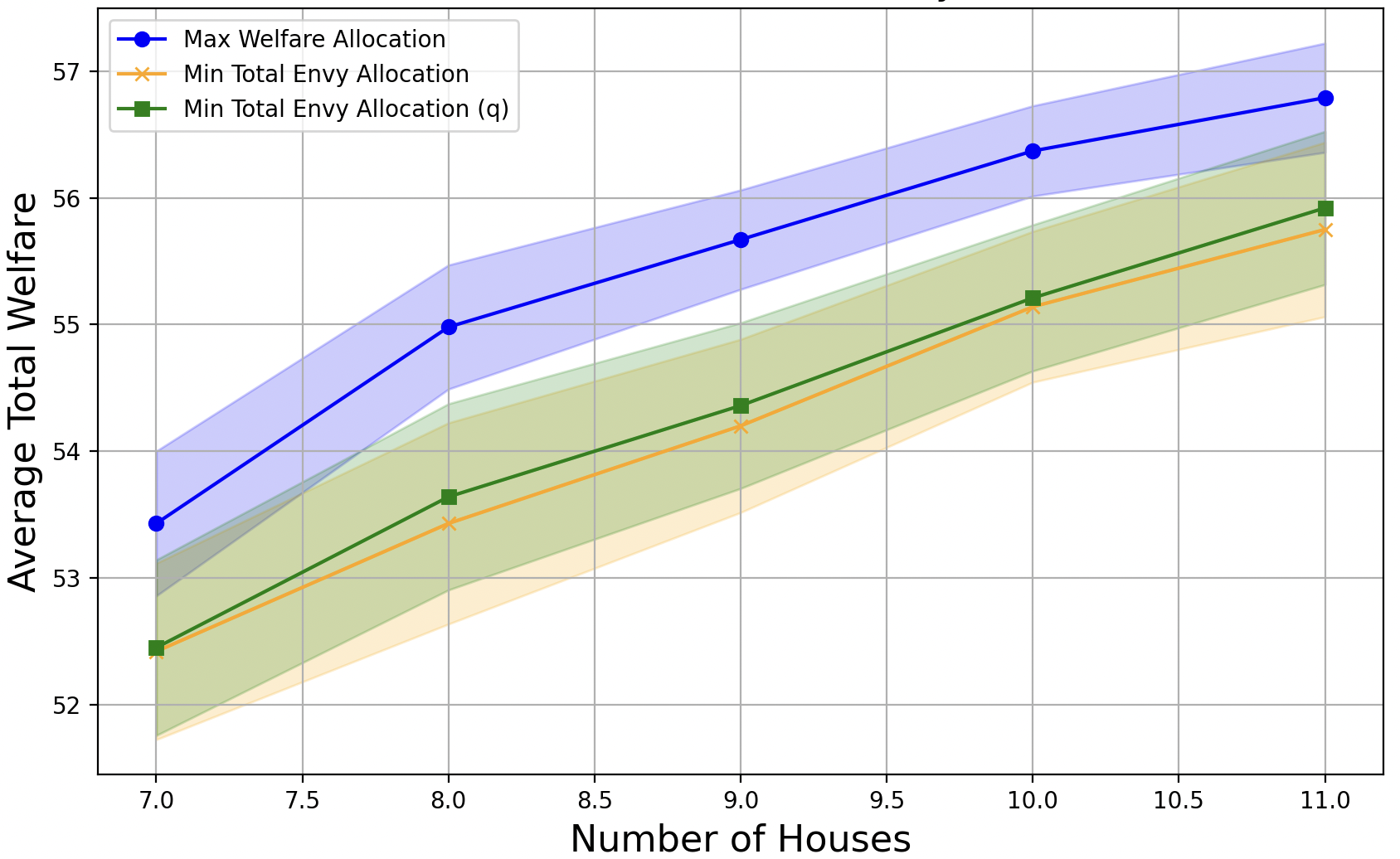}
    \caption*{$q= 3$}
  \end{minipage}
  
  \vspace{0.5cm}
  
  \begin{minipage}[t]{0.45\textwidth}
    \centering
    \includegraphics[width=\linewidth]{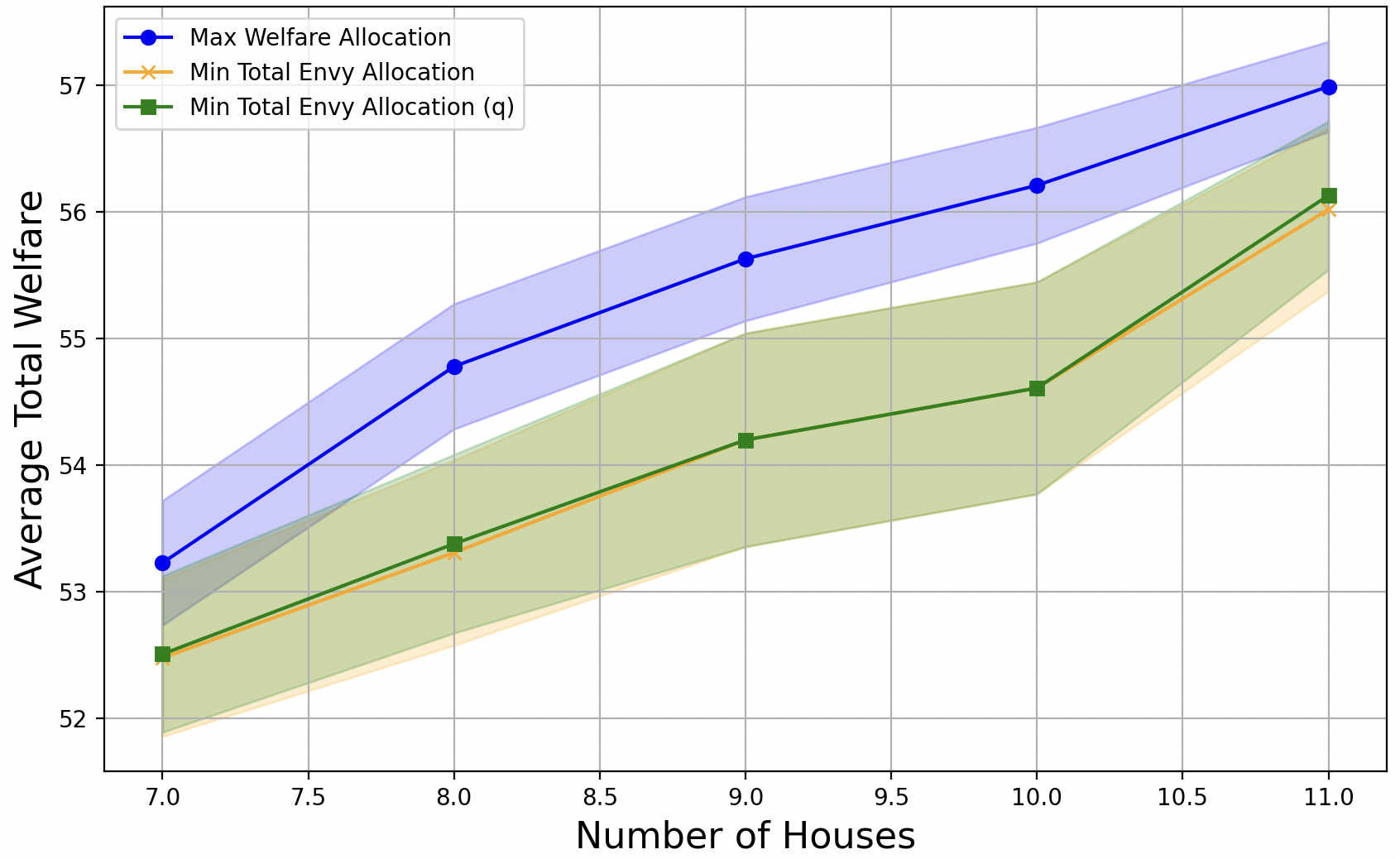}
    \caption*{$q=4$}
  \end{minipage}
  \hfill
  \begin{minipage}[t]{0.45\textwidth}
    \centering
    \includegraphics[width=\linewidth]{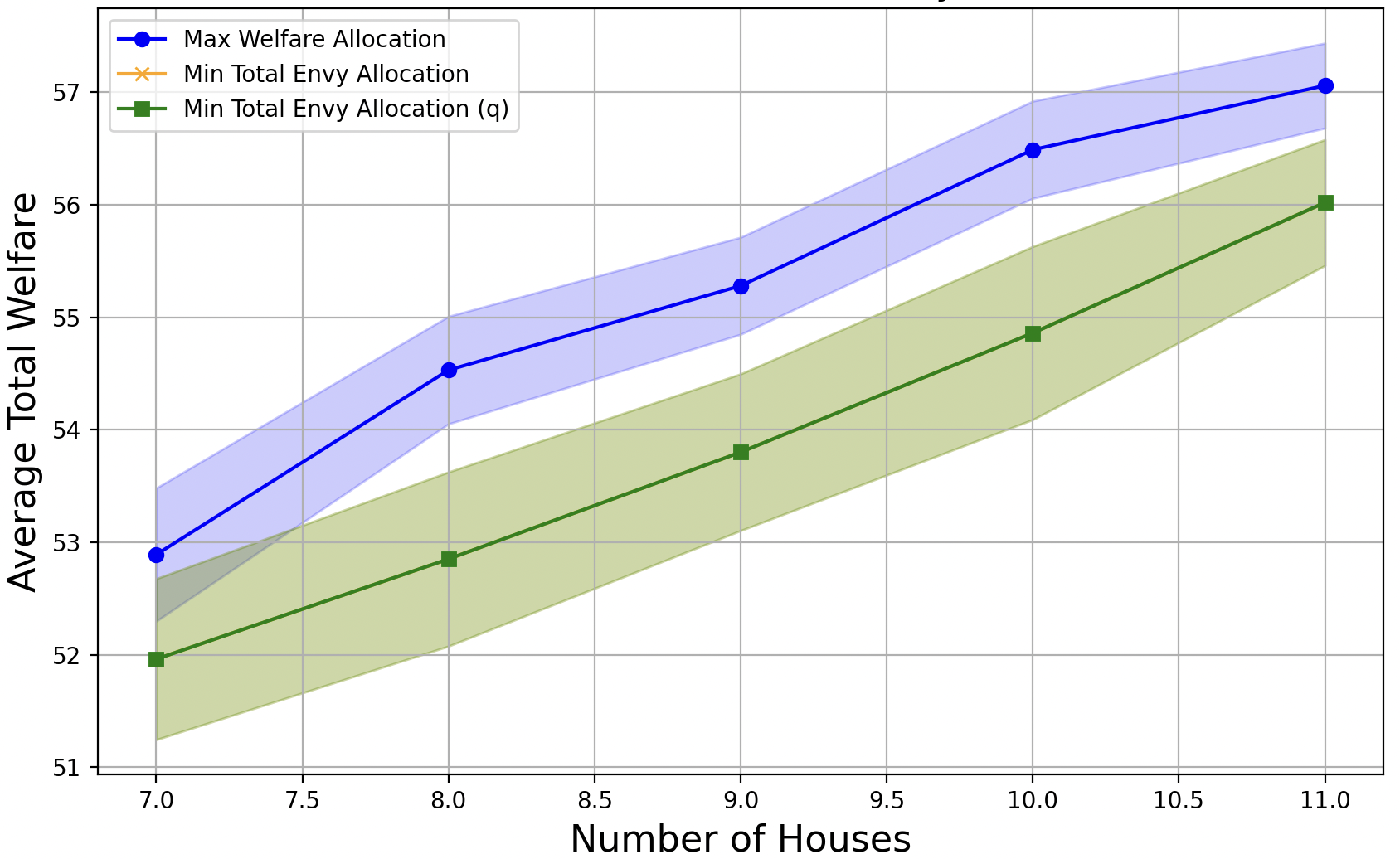}
    \caption*{$q= 6$}
  \end{minipage}
  
  \caption{Welfare loss incurred in allocations that minimize $\tenvy$ with at most $q$ reallocations.}
  \label{fig:fourfigs_total_envy}
\end{figure*}

\end{document}